\documentclass[letter,12pt]{article}
\usepackage[T1]{fontenc}
\usepackage[dvips]{graphicx}
\graphicspath{{images/}}

\usepackage{algorithm}
\usepackage[noend]{algpseudocode}
\def\BState{\State\hskip-\ALG@thistlm}

\usepackage{verbatim}
\usepackage{textcomp}
\usepackage{amsmath,amsthm}
\usepackage{xspace}
\usepackage{pifont}
\usepackage{graphicx}
\usepackage{amssymb}
\usepackage{epic, eepic}
\usepackage{dsfont}
\usepackage{amssymb}
\usepackage{makeidx}
\usepackage{mathrsfs}
\usepackage{exscale}
\usepackage{color} 
\usepackage{overpic} 
\usepackage{bm}
\usepackage{bbm}
\usepackage{booktabs} 
\usepackage{color, colortbl}
\usepackage{subcaption}
\usepackage{pgfplots}

\definecolor{Gray}{gray}{0.9}

\usepackage{amsmath,afterpage}
\usepackage{epsf}
\usepackage{graphics,color}

\def\0{\mathbf{0}}

\def\eps{\varepsilon}

\def\lam{\lambda}
\def\rr{\rightarrow}

\def \< {\langle}
\def \> {\rangle}

\def\erf{\text{erf}}

\def\beqa{\begin{eqnarray}}
\def\eeqa{\end{eqnarray}}
\def\beqas{\begin{eqnarray*}}
\def\eeqas{\end{eqnarray*}}

\newtheorem{theorem}{Theorem}[section]
\newtheorem{lemma}[theorem]{Lemma}

\newtheorem{proposition}[theorem]{Proposition}
\newtheorem{conjecture}[theorem]{Conjecture}
\newtheorem{corollary}[theorem]{Corollary}

\newtheorem{remark}[theorem]{Remark}

\newtheorem{definition}[theorem]{Definition}

\numberwithin{equation}{section}
\newcommand{\hatd}[1]{{}}

\newcommand{\bd}{\begin{displaymath}}
\newcommand{\ed}{\end{displaymath}}
\newcommand{\be}{\begin{equation}}
\newcommand{\ee}{\end{equation}}
\newcommand{\bq}{\begin{eqnarray}}
\newcommand{\eq}{\end{eqnarray}}
\newcommand{\bn}{\begin{eqnarray*}}
\newcommand{\en}{\end{eqnarray*}}

\newcommand{\re}{\mathds{R}}

\def \E {{\mathbb{E}}}

\DeclareMathOperator*{\argmin}{arg\,min} 
\newcommand{\sign}{\textrm{sign}}

\usepackage{authblk}

\title{Phase Transitions in Kyle's Model with Market Maker Profit Incentives}
\author{Charles-Albert Lehalle\thanks{Capital Fund Management and CFM-Imperial College Institute}, Eyal Neuman\thanks{Department of Mathematics, Imperial College London and CFM-Imperial College Institute}, Segev Shlomov\thanks{Faculty of Industrial Engineering and Management, Technion - Israel Institute of Technology}}

\begin{document}

 \vspace{-0.5cm}
\maketitle

\begin{abstract}
We consider a stochastic game between three types of players: an inside trader, noise traders and a market maker. In a similar fashion to Kyle's model, we assume that the insider first chooses the size of her market-order and then the market maker determines the price by observing the total order-flow resulting from the insider and the noise traders transactions. In addition to the classical framework, a revenue term is added to the market maker's performance function, which is proportional to the order flow and to the size of the bid-ask spread. We derive the maximizer for the insider's revenue function and prove sufficient conditions for an equilibrium in the game. Then, we use neural networks methods to verify that this equilibrium holds. We show that the equilibrium state in this model experience interesting phase transitions, as the weight of the revenue term in the market maker's performance function changes. Specifically, the asset price in equilibrium experience three different phases: a linear pricing rule without a spread, a pricing rule that includes a linear mid-price and a bid-ask spread, and a \emph{metastable} state with a zero mid-price and a large spread.

\end{abstract}

\begin{description}
\item[Keywords:] market making, price impact, market equilibrium, metastability,
  Kyle's model, market liquidity, market microstructure, neural networks  
\end{description}

\bigskip

\section{Introduction}


The concepts of market liquidity, price impact, information asymmetry and adverse selection have always been at the center of market microstructure research. The seminal paper by Kyle \cite{kyle85} made connections between all these concepts in a simple and tractable framework and became a corner stone in the literature of this field. In his model, Kyle described a game between three types of players: an inside trader (insider), noise traders and a market maker. A risky asset is traded over one period, where the insider has an exclusive information on the price position at the end of the term. This is often referred to as the fundamental price. Based on this information, the insider decides on the size of her market-order. At the same time the noise traders also submit their market-orders without any information about the price dynamics, so the total size of their orders is modelled as a centred random variable. The sum of all these orders, which is the order-flow, is then revealed to the market maker without the possibility to disentangle its components. Based on this observation, the market maker decides on the mid-price of the asset, and clears the orders. The presence of noise traders helps the insider to obscure her position from the marker maker. Under Gaussian assumptions on the distribution of the fundamental price and the noise traders orders, Kyle proved that this game has an equilibrium, in which the insider's strategy is linear with respect to the fundamental price, and the market maker's pricing rule is linear with respect to the orderflow. Kyle also extended the proof to a multi-period version of this model.  

The simple setting of Kyle's model reveals some fundamental connections between key concepts in market microstructure. In equilibrium the insider adjusts the order size according to the fundamental price, while taking into account the price impact of her order. Moreover, dependence in the parameters of the distribution of the fundamental price, is reflected in the market maker's pricing. This provides insights on the effect of asymmetric information, and more specifically adverse selection, in pricing strategies. 

Numerous extensions to Kyle's model were studied, we briefly survey just a few of them. 
Subrahmanyam \cite{Subrahmanyam:1991aa} studied an extension to Kyle's model where both the informed trader and market maker are risk averse. Nishide \cite{Nishide2006InsiderTW} investigated a version of the model with competing market makers. Boulatov and Bernhardt \cite{Boulatov:2015aa} considered the robustness of the linear Kyle equilibrium with respect to small perturbations in the payoffs of the agents. Molino et al. \cite{Garcia-del-Molino:2020aa} studied the case where the market maker is setting the price of $n$ correlated securities. A neural networks approach to Kyle's single period model was developed by Friedrich and Teichmann \cite{Teichmann20}. They showed that the agents strategies converges to the linear equilibrium, which was proved for the Gaussian case by Kyle, also for various types of fundamental price distributions. A continuous time version of Kyle's model was first proposed by Back \cite{Back:1992aa}. Collin-Dufresne and Fos \cite{Collin16} extended Back's work to the case where the liquidity provided by noise traders follows a general stochastic process. Significant amount of work on the mathematical foundations of the continuous Kyle model in the context of filtering, enlargement of filtrations and Markov bridges, is described in the lecture notes by \c{C}etin \cite{cetin-notes} and references therein.

The main purpose of market makers is to add liquidity to markets by being ready to buy and sell assets at any time during the trading day. As a result, market makers also determine the spread between the bid and ask (i.e. the difference between the price quotes for market buy and sale orders) and even if it is only a few cents, they can profit by executing thousands of trades in a single day. None of the extensions of Kyle's model that were mentioned earlier take into account the fact that market makers also decide on the bid-ask spread and their profits depend on this decision. The market-maker spread is also considered a measure for assets liquidity, spreads tend to be tighter in more actively traded assets, and in those that have more available market makers. The size of the spread is also one of the main components of traders transaction costs. A related work by El Euch et al. \cite{rosen2020}, proposed a model for an exchange (or a regulator) who is aiming to attract liquidity to the market. The exchange was looking for the best make–take fees policy to offer to market-makers in order to maximise its utility. 

As mentioned earlier, Market makers earn money by having investors and traders buy assets in the ask price and sale them assets on a lower bid price.  The wider the spread, the more potential profit the market maker can make. On the other hand, the competition among market makers can keep spreads tight. Therefore, a key addition to Kyle's model is to introduce the revenue of the market maker due to the spread and to capture the trade-off between providing a comparative price and earning money from the spread. We include in our model both the market maker's revenue along with the decision on the size of the spread, as described in Section \ref{sec-model}.  
 We then derive the maximizer of the informed trader's revenue and give sufficient conditions for equilibrium in the game. We use neural networks methods to verify that indeed this equilibrium holds, and show that it experience phase transitions as we increase the relative weight of the revenue term with respect to the price efficiency in the market maker performance function.  As presented in Figure \ref{fig-phase}, the equilibrium price in this game has three phases, the "Kyle phase" where the spread is zero and the mid-price is linear, the "linear mid-price with spread phase", and the "spread phase" where the market maker does not use any price rule other than the bid-ask spread (see also Figure \ref{phase2}).  
\paragraph{Organization of the paper:} 
 In Section \ref{sec-model}  we define a new extension to Kyle's model that takes into account the market maker's revenue from creating a spread along with price efficiency. In Section \ref{sec-res} we present our main results that include the existence of a unique solution to the insider's optimization problem and gives sufficient conditions for the equilibrium in the system. We also provide a neural network algorithm that solves the market market maker's optimization problem and hence derives the equilibrium. In the end of this section we prove the existence of a metastable equilibrium, that derived in a closed form. In Section \ref{section-neural} we provide a detailed description on the neural networks methodology.  Sections \ref{sec-prf}--\ref{sec-pfs-3} are dedicated to the proofs of the main theoretical results. Finally in Section \ref{sec-form} we give some explicit formulas for the equilibrium points of the game.

 \section{The Model} \label{sec-model} 
We consider a one-period model that consists of three types of agents:  an informed trader (insider), noise traders and a market maker.
We assume that the future price (or fundamental price) at the end of the period, is predicted by the informed trader, and it is a random variable $\tilde v$ with a mean $p_{0}$ and variance $\sigma_{\tilde v}^2$. The noise traders have no predictions on the price move, and we denote by $\tilde u$ the total amount that they trade, which is a symmetric random variable with variance $\sigma_{\tilde u}^2$ and with a continuous probability density function $f_{\tilde u}$. It is assumed that $\tilde u$ and $\tilde v$ are independent random variables. Finally, we denote by $\tilde x$ the amount traded by the insider and by $\tilde p$ the execution price which is determined by the market maker . 

As in \cite{kyle85} we describe the trading as a two steps procedure. First, the values of $\tilde v$ and $\tilde u$ are realized and the insider chooses the size of her market order $\tilde x$. Note that when choosing $\tilde x$, the insider 
knows $\tilde v$ but not $\tilde u$. We define $\tilde x = X(\tilde v)$, where $X$ is a measurable function. In the second step, the market maker determines the traded (or execution) price $\tilde p$, while observing only the total order-flow $\tilde x+\tilde u$. Our main objective is to reflect the revenue of the market maker in her performance function. Since these revenue are directly linked to the bid-ask spread, we enlarge the class of linear prices which was proposed by Kyle. Motivated by Madhavan et al. \cite{madhav96} we assume that $\tilde p = P(\tilde x+\tilde u)$, where the price function $P$ is in the class of functions:
\be \label{admis-p} 
\mathscr{P} = \{P(x)=\lambda x + \theta\, \textrm{sign}(x)+p_{0}, \, \lambda, \theta > 0\}. 
\ee 

\begin{remark} 
The choice of $\mathscr{P}$ in \eqref{admis-p} is the simplest way to define a price with a symmetric bid-ask spread, where the size of the spread is $\theta$. Our choice is consistent with Madhavan et al. \cite[Section 2]{madhav96}, where we note that the conditional expectation of the indicators in their one period model, can be replaced by the sign of $\tilde x+ \tilde u$.  In Section \ref{sec-en-admis} we provide numerical evidence that extending the class $\mathscr{P}$ by adding higher order terms will still lead to equilibrium  price in $\mathscr{P}$. 
\end{remark} 

The profit of the insider, $\tilde \pi$ is given by $\tilde \pi = (\tilde v -\tilde p)\tilde x$. Note that $\tilde \pi= \tilde \pi(X,P)$.

\begin{definition}[Equilibrium]
An equilibrium between the market maker and the insider is a pair of $X$ and $P$ such that the following two conditions hold. 
\begin{itemize} 
\item[\textbf{(i)}] \textbf{Profit maximization}: for any other strategy $X'$ and for any $v \in \re$, 
\be \label{trade-opt} 
\mathbb E\big[ \tilde \pi(X,P) | \tilde v =v  \big] \geq \mathbb E\big[ \tilde \pi(X',P) | \tilde v =v  \big]. 
\ee
\item[\textbf{(ii)}]  \textbf{Market Efficiency and Revenue}: the random price $\tilde p =P(\tilde x+\tilde u)$ satisfies 
\be \label{equi-mm}
\argmin_{P\in \mathscr{P}}    \big\{   \mathbb E\big[ (\tilde v - \tilde p)^{2} \big] - \gamma\theta E[|x(\tilde v) + \tilde u|] \big\}, 
\ee
where $\gamma>0$ is a fixed risk-aversion constant.  
\end{itemize} 
\end{definition}

\begin{remark} 
In Kyle's paper \cite{kyle85} the market maker's efficiency criterion was given by 
\be \label{eff-kyle} 
 \tilde p=P(\tilde x+\tilde u) = \E[\tilde v| \tilde x+\tilde u ].
 \ee
Note that in the setting of Theorem 1 in \cite{kyle85}, minimising $\mathbb E\big[ (\tilde v - \tilde p)^{2}\big]$, which is our \emph{market efficiency and revenue} criterion (\ref{equi-mm}) with $\gamma=0$, is equivalent to (\ref{eff-kyle}). In our model we incorporate the revenue of the market maker, so we add the term $\theta E[|\tilde v + \tilde u|]$, which reflects the revenue, as it is proportional to the size of the spread and to the total overflow. In Proposition \ref{prop-lin-eq} we prove that this term is essential in order to get a difference between the buy and sell prices. We call $\gamma$ the risk-aversion parameter since it describes the tradeoff between keeping an efficient price and making profits. The second clearly may create additional risk by suppressing insider from trading large orders and therefore trading in other venues.  
\end{remark}

\section{Main Results} \label{sec-res}
We first present our theoretical results, where we solve the trader's profit maximisation problem. We also give a necessary condition for finding the equilibrium. Then we will construct a neural network that will allow us to numerically find the equilibrium. We also prove the existence of a metastable equilibrium and derive it in a closed form.

\subsection{Solution of the insider's problem} 
In the next proposition we show the existence of an optimal strategy for the insider. We also provide some insights on the properties of this strategy. We recall that the fundamental price at the end of the period $\tilde v$ is a random variable with a mean $p_{0}$ and variance $\sigma_{\tilde v}^2$. The noise traders order flow $\tilde u$ is a symmetric random variable with variance $\sigma^2_{\tilde u}$ and a continuous probability density function $f_{\tilde u}$. We denote by $F_{\tilde u}$ the cumulative distribution function  of $\tilde u$.
Moreover, it is assumed that $\tilde u$ and $\tilde v$ are independent random variables. Note that at this point we do not specify the distributions of $\tilde v$ and $\tilde u$.  We postpone the proofs of all the theoretical results of this section to Section \ref{sec-prf}.


\begin{proposition} \label{lemma-min-trader} 
For any $v\in \re$ there exists a unique $x^{*}= x^{*}(v)$ that maximizes that expected profit of the insider  
\be \label{inside-prob} 
R_{v}(x) := \mathbb E\big[ \tilde \pi(X,P) | \tilde v =v  \big]. 
\ee
The maximizer $x^{*}$ satisfies the following properties:
\begin{itemize} 
\item When $v=p_{0}$ we have $x^{*}=0$ and $R_{p_{0}}(x^{*})=0$. 
\item When $v \not =p_{0}$, then $x^{*}(v)$ is a solution to the equation, 
\be \label{deriv-cond} 
\theta F_{\tilde u}(-x) -x\big(\theta f_{\tilde u}\big(-x\big)-\lambda\big) = \kappa(v), 
\ee 
where $\kappa(v) =(v-p_{0}-\theta)/2$. We moreover have $R_{v}(x^{*}(v)) >0$ and $\sign(x^{*}(v)) = \sign(v-p_{0})$.
\end{itemize} 
 \end{proposition}
 
 \begin{remark} 
The proof of Proposition \ref{lemma-min-trader} suggests that if 
\be \label{2nd-cond} 
\frac{d^2}{dx^2}(xF_{\tilde u}(-x)) < 0, \quad \textrm{for all } x \in \re \setminus\{0\},
\ee
 then $R_v$ is concave and (\ref{deriv-cond}) has a unique solution. An example for that is when $\tilde u$ has a centred Laplace distribution. 
 \end{remark}

 \begin{remark} 
 Note that in the case where there is no bid-ask spread, i.e. $\theta =0$, we recover the result of Theorem 1 in \cite{kyle85} and get that 
\be \label{x-no-sprd} 
x^{*}(v) = \frac{v-p_{0}}{2\lam}. 
\ee
 \end{remark}

 \subsubsection{Solution to the Gaussian noise case}
 We specialise in the case where the total order flow of the noise traders $\tilde u$ is a mean-zero Gaussian random variable. Denote by $\Phi$ (respectively $\phi$) the cumulative distribution function (respectively the probability density function) of a standard Gaussian.
 
 \begin{corollary} 
Assume the same hypothesis as in Proposition \ref{lemma-min-trader}, only now let $\tilde u$ be a mean-zero Gaussian with variance $\sigma^2_{\tilde u}$. Then (\ref{deriv-cond}) is given by 
 \be \label{normal-eqn} 
\theta\Phi\big(-\frac{x}{\sigma_{\tilde u}}\big) -x\Big(\frac{\theta}{\sigma_{\tilde u}}\phi\big(-\frac{x}{\sigma_{\tilde u}}\big)-\lambda\Big) = \kappa(v).
\ee
 \end{corollary}
 
The following Lemma characterises the global maximum for the informed trader problem under the Gaussian noise assumption.
\begin{proposition}\label{gaus}
Let $\tilde u$ be a mean-zero Gaussian random variable, then there are two possible cases: 
\begin{enumerate} 
\item[\textbf(a)] there exists a unique one solution $x^*$  to (\ref{normal-eqn}) and this is the maximizer of $R_{v}$. 
\item[\textbf(b)] there exist three solutions $x_1^*<x^*_2<x_3^*$ to (\ref{normal-eqn}), and the global maximizer of $R_{v}$ is either $x_1^*$ or $x^*_3$. 
\end{enumerate} 
\end{proposition}
 
 \subsubsection{Solution to the Uniform noise case} \label{unif-equi} 
We study in greater detail the case where $\tilde u$ is a Uniform random variable on $[-1,1]$.  \begin{proposition} \label{prop-unif} 
Assume the same hypothesis as in Proposition \ref{lemma-min-trader}, only now let $u$ be Uniform on $[-1,1]$. The unique  maximizer $x^{*}= x^{*}(v)$ which maximizes the expected profit of the trader in \eqref{inside-prob} is given by 
\begin{itemize} 
\item[\bf{(i)}] $x^*(v) = \frac{v-p_0}{2(\lambda +\theta)}$, for $0<v-p_0 \leq \lambda +\theta +\sqrt{(\lambda +\theta)\lambda}$,
\item[\bf{(ii)}] $x^*(v) = \frac{v-p_0-\theta}{2\lambda }$, for $v-p_0>\lambda +\theta +\sqrt{(\lambda +\theta)\lambda}$.
\end{itemize} 
\end{proposition}

\subsection{Sufficient conditions for equilibrium}  
In this section we provides sufficient conditions for the existence of an equilibrium. We continue to assume that the future price $\tilde v$ is a random variable with a mean $p_{0}$ and variance $\sigma_{\tilde v}^2$ and that the noise traders order flow $\tilde u$ is a symmetric random variable with variance $\sigma_{\tilde u}^2$ and continuous density $f_{\tilde u}$. Also here we do not specify the distribution of $\tilde v$ and $\tilde u$. The proofs of the theoretical results in this section are given in Section \ref{sec-pfs2}. 
 
We consider $x^*(v)$ from Proposition \ref{lemma-min-trader} which is the maximizer of the insider's expected profit \eqref{inside-prob}. Before stating our main result, we introduce the following notation. Let 
\bn
\ell_{p,x^{*}}&=& \E\big[|x^{*}(\tilde v)+\tilde u|^{p}\big], \quad p=1,2,  \\
\mu_{x^{*}} &=&\E\big[x^{*}(\tilde v)( \tilde v- p_{0})\big], \\
\kappa_{x^{*}} &=&\E\big[ \textrm{sign}(x^{*}(\tilde v)+\tilde u)( \tilde v- p_{0})\big]. 
\en
Were often write $\ell_{p}, \mu, \kappa$ to simplify the notation. Note that $\ell_{p,x^{*}}, \mu_{x^{*}}, \kappa_{x^{*}}$ are all functions of $(\lambda, \theta)$.

In the next theorem we characterise the equilibrium between the market maker and the insider. 
\begin{theorem} [sufficient condition]    \label{thm-equil1} Assume that $\tilde v$ is a random variable with a mean $p_{0}$ and variance $\sigma_{\tilde v}^2$ and that $\tilde u$ is symmetric random variable with variance $\sigma_{\tilde u}^2$ and a continuous density. 
For any $x^{*}(v)$, which is given in Proposition \ref{lemma-min-trader}, if the following system 
\be \label{lam-th-eq} 
\lam  = \frac{\mu -(\kappa +\gamma \ell_{1}/2) \ell_{1}}{\ell_{2}-\ell_{1}^{2}}, \quad \theta = \kappa+\frac{\gamma}{2}\ell_{1}- \frac{\ell_{1}(\mu -\kappa \ell_{1}-\ell_{1}\gamma/2)}{\ell_{2}-\ell_{1}^{2}},
\ee
has a non-negative solution $(\lambda^*, \theta^*)$, then the optimal price that minimizes the market maker's objective function \eqref{equi-mm} is given by 
$$P^*(x)=\lambda^{*} x + \theta^{*} \textrm{sign}(x)+p_{0}.$$  
Moreover, $(x^{*}(\cdot), \lambda^*, \theta^*)$ is an equilibrium of the game. 
\end{theorem} 


\begin{remark} 
 Note that finding a solution to equation (\ref{lam-th-eq}) is a difficult task since $\ell_i, \kappa$ and $\mu$ depend on $(\lam,\theta)$. 
In Section \ref{section-find-eq} we provide a numerical method, based on an ad hoc neural network,
that can find the equilibrium point $(\lam^*,\theta^*)$. 
Proving the uniqueness of the equilibrium seems to be out of reach due to the complexity of (\ref{lam-th-eq}).
Nevertheless, our numerical approach provides evidence that uniqueness indeed holds.
\end{remark} 

\begin{remark} 
We observe that in the case where we restrict to pricing rules with $\theta =0$ (i.e. zero spread), $\tilde u \sim N(0,\sigma_{\tilde u}^{2})$ and $\tilde v\sim N(p_{0}, \sigma_{\tilde v}^{2})$, then $x^{*}(v)$ is given by (\ref{x-no-sprd}), 
$\kappa = \frac{\ell_{1}\mu}{\ell_{2}}$ and 
$$\lam^* =\frac{\mu}{\ell_{2}} =\frac{\beta\sigma_{\tilde v}^{2}}{\beta^{2}\sigma^{2}_{\tilde v}+\sigma_{\tilde u}^{2}},$$
where $\beta = \frac{1}{2\lam}$. It follows that $P^{*}(x)$ is similar to the price at equilibrium in Theorem 1 of \cite{kyle85}. 
\end{remark}

In the following proposition we prove that when $\gamma$ in \eqref{equi-mm} is set to zero, we recover the classical Kyle equilibrium without a bid-ask spread (i.e. $\theta =0$). 

\begin{proposition}  \label{prop-lin-eq} 
Assume that $\tilde v - p_{0}$ is a centred Gaussian with variance $\sigma_{\tilde v}^2$ and that $\tilde u$ is either a centred Gaussian or centred Uniform with variance $\sigma_{\tilde u}^{2}$. If the risk-aversion parameter $\gamma$ in \eqref{equi-mm} is zero, then there exists an equilibrium in which $X$ and $P$ are linear functions that are given by 
$$
X(v) = \beta^{*}(v-p_{0}) \quad P(x) = p_{0}+ \lam^{*} x, 
$$
where $\beta^{*} = \frac{\sigma_{ u}}{\sigma_{ v}}$ and $\lam^{*} =   \frac{\sigma_{ v}}{2\sigma_{ u}}$. 
\end{proposition}

\subsection{Numerical results: finding the equilibrium} \label{section-find-eq}
In this section we find the equilibrium points of the game under the assumption that $\tilde v - p_{0}$ is a centred Gaussian with variance $\sigma_{\tilde v}^{2}$ and that $\tilde u$ is either a centred Gaussian with variance $\sigma_{\tilde u}^{2}$ or Uniform on $[-1,1]$. In order to derive the equilibrium we design an ad hoc neural network, which is
described in detail in Section \ref{section-neural}. 

In figure \ref{gamma-unif} we plot the optimal $\lambda^{*}$ and $\theta^{*}$ as a function of the risk-aversion parameter $\gamma$ for the Gaussian (left panel) and Uniform (right panel) cases. We also show: the expected insider's optimal market-order size, her optimal revenue, and the market maker value function as a function of the risk-aversion parameter.
As expected in both cases, when the risk-aversion parameter increases, $\lambda^*$ decreases and the size of the spread $\theta^*$ increases. In addition, we observe that when $\gamma$ increases, the market maker gives more weight to revenue, and the insider increases the order size. This however does not necessarily implies an increase of the revenue. From the market maker point of view, we observe a logarithmic increase in its performance function \eqref{equi-mm}, as as $\gamma$ increases.  This is due to the increase in the trader's order size, along with the increase in the revenue made from the spread. 

\begin{figure} [h!]  
  \begin{subfigure}[b]{0.5\textwidth}
    \includegraphics[width=\textwidth, trim=12mm 0 10mm 0, clip]{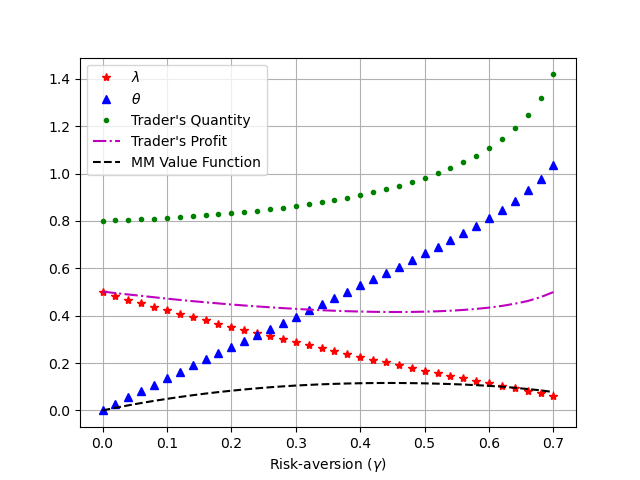}
      \end{subfigure}
 \hfill
  \begin{subfigure}[b]{0.5\textwidth}
    \includegraphics[width=\textwidth,  trim=12mm 0 10mm 0, clip]{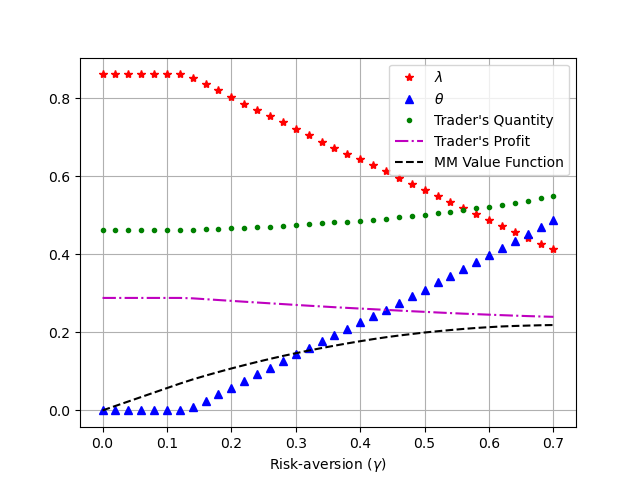}
     \end{subfigure}
    \caption{Plot of the equilibrium price parameters $\lambda^*$ (red) and $\theta^*$ (blue) as a function of the risk aversion $\gamma$. We also show the expected insider's transaction size (green), expected insider profit (purple) and the market maker performance functional (black). The Gaussian noise case is presented on the left panel and the Uniform noise case on the right panel. 
}
   \label{gamma-unif} 
\end{figure}

In figure \ref{oreder-fig} we fix $\gamma=0.5$ and plot the insider's optimal order size, revenue and the corresponding total order flow as a function of the price $v$ at equilibrium, both in the Gaussian and Uniform cases. Following our theoretical results, we observe that the optimal trade size $x^{*}(v)$ is symmetric with respect to $v$, and it is nonlinear in $v$. We observe that in both cases, when the the future price $|\tilde v|$ is roughly larger than one, the insider trades more aggressively, even if her position is detected by the market maker. Similar plot is presented in figure \ref{oreder-fig2} for the cases where $\gamma =10$ and $\gamma =16$. 

\begin{figure} [h!]  
  \begin{subfigure}[b]{0.5\textwidth}
     \includegraphics[width=1\textwidth, trim=10mm 0 10mm 0, clip]{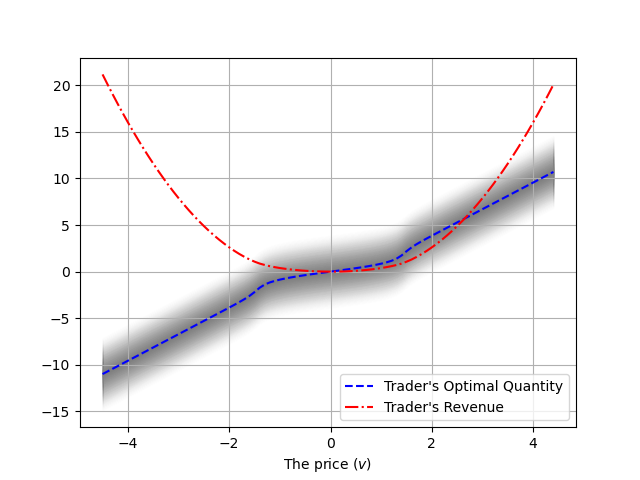}
    \caption{Gaussian Noise}
      \end{subfigure}
 \hfill
  \begin{subfigure}[b]{0.5\textwidth}
    \includegraphics[width=1\textwidth, trim=10mm 0 10mm 0, clip]{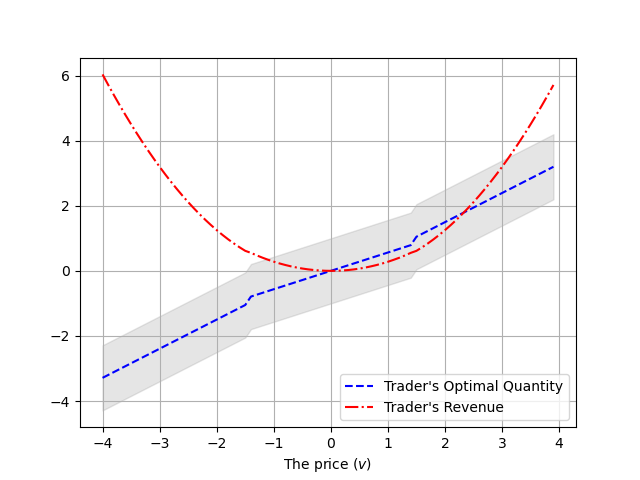}
    \caption{Uniform Noise }
      \end{subfigure}
      \caption{$\gamma=0.5$: the insider expected revenue (red) and transaction size $x^{*}$ (blue) in equilibrium ($y$-axis) vs. $v$ (in the $x$-axis) in the "linear mid-price with a bid-ask spread phase". The total order flow $v+\tilde u$ appears grey spectrum.
      }\label{oreder-fig}
\end{figure}

\begin{figure} [h!]  
  \begin{subfigure}[b]{0.5\textwidth}
     \includegraphics[width=1\textwidth, trim=10mm 0 10mm 0, clip]{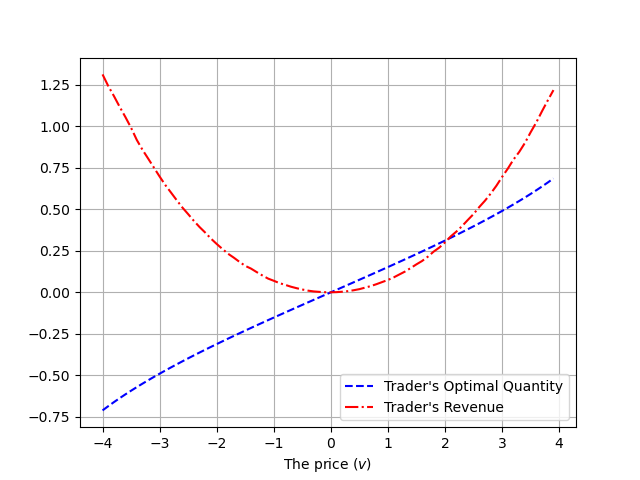}
    \caption{Gaussian Noise}
       \end{subfigure}
 \hfill
  \begin{subfigure}[b]{0.5\textwidth}
    \includegraphics[width=1\textwidth, trim=10mm 0 10mm 0, clip]{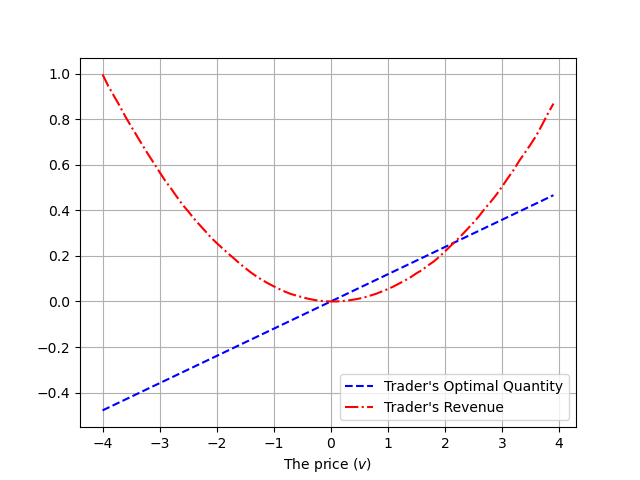}
    \caption{Uniform Noise }
      \end{subfigure}
       \caption{$\gamma=10$ (left) and $\gamma=16$ (right): the insider expected revenue (red) and transaction size $x^{*}$ (blue) in equilibrium ($y$-axis) vs. $v$ (in the $x$-axis) in the "bid-ask spread phase". 
       }\label{oreder-fig2}
\end{figure}

We discuss the effect of the risk aversion factor $\gamma$ on the type of the equilibrium in the model. As Proposition \ref{prop-lin-eq} suggests when $\gamma =0$ we have the classical Kyle equilibrium without a bid-ask spread. We can numerically show that both for Gaussian noise and Uniform noise, equilibrium exists for an interval of positive $\gamma$'s. More precisely, there exists $\gamma_{LBid}>0$ such that for every $0<\gamma < \gamma_{LBid}$ the price in equilibrium is of the form $P^{*}(x) = \lambda^{*}x + \theta^{*}\,\sign(x)$ where both $\theta^{*}$ and $\lambda^{*}$ are positive. For the Gaussian case $\gamma_{LBid} \approx 0.7$ and for the Uniform case $\gamma_{LBid} \approx 1$. 
Moreover there exists $\gamma_{LBid}  <\gamma_{Bid}$ such that for any $\gamma_{LBid} <\gamma  < \gamma_{Bid}$  equilibrium doesn't hold, where in the Gaussian case $\gamma_{Bid} \approx 10$ and in the Uniform case $\gamma_{Bid} \approx 16$. Finally in the third phase, where $\gamma> \gamma_{Bid}$ we have an equilibrium with a bid-ask spread only, namely $\lambda^{*}=0$ and $\theta^{*}>0$. The equilibrium in the this regime a \emph{metastable state}, since once the algorithm arrives to equilibrium, it will not get out with probability asymptotically close to one. However, it does not admit the classical definition of equilibrium point. We state and prove the precise result on the metastable equilibrium in  Section \ref{sec-meta}. These results are summarised in Figures \ref{fig-phase} and \ref{phase2}. 
 
 \begin{figure}[h!]
    \centering
    \includegraphics[width=\textwidth]{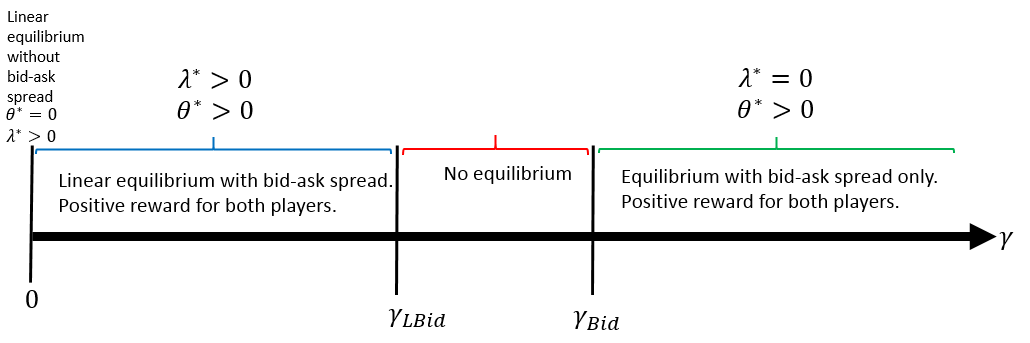}
    \caption{Equilibrium phase transitions. When $0<\gamma<\gamma_{LBid}$ the equilibrium price is $P^{*}(x)=\lambda^{*} x + \theta^{*} sign(x)$. For $\gamma>\gamma_{Bid}$ the equilibrium price function is  $P^{*}(x)= \theta^{*} sign(x)$. When $\gamma_{LBid}<\gamma<\gamma_{Bid}$ no equilibrium was found. }
    \label{fig-phase}
\end{figure}

\begin{figure}[h!]
    \centering
    \includegraphics[width=\textwidth]{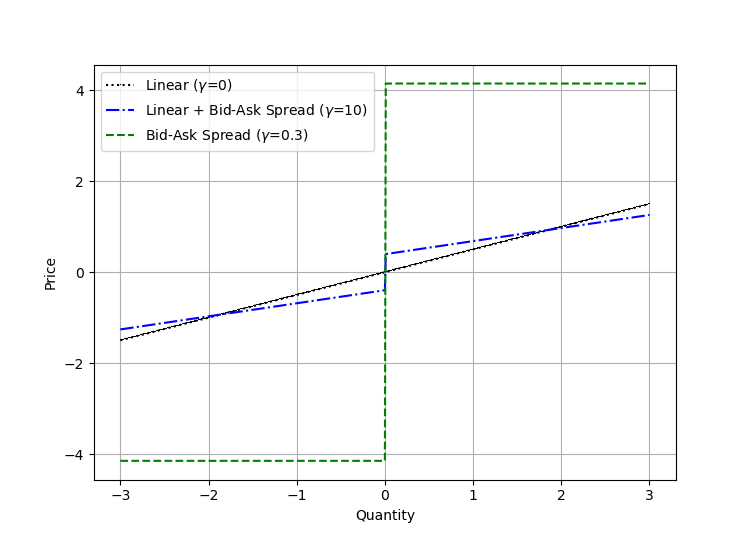}
    \caption{Plot of the price at equilibrium for the Gaussian case at the three different phases: for $\gamma =0$ (black), $\gamma =0.3$ (blue)  and $\gamma =10$ (green).}
    \label{phase2}
\end{figure}

\subsection{Existence of a metastable equilibrium}  \label{sec-meta}
In this section we state and prove the the precise results on the metastable equilibrium which was found numerically in Section \ref{section-find-eq} for $\gamma >\gamma_{Bid}$. In order to define this equilibrium we first state our search algorithm for the equilibrium. 

We first demonstrate the algorithm to the classical Kyle's model where $\gamma$ in \eqref{equi-mm} is set to zero, and therefore the equilibrium price has a zero bid-ask spread. 

  \begin{algorithm}[H]
\caption{Equilibrium price for the classical Kyle model}\label{alg-kyle-0}
\begin{algorithmic}[1]
\State Initialise the price function $P(x) = \lam_0x+b_0$ with some arbitrary weights $\lam_0, b_0$.
\State
Given that $P(x) = \lam_nx+b_n$, find the function $x_n(v)$  which maximize the trader's optimization problem: 
\bd
 R_{v}(x)  =\mathbb E\big[ (\tilde v -(\lam_n x_n(\tilde v)+b_n))\tilde x_n(\tilde v) | \tilde v =v  \big].
\ed

\State
Find $\lam_{n+1}, b_{n+1}$ which minimizes market maker's cost function:
$$
 E\big[ \big(\tilde v - (\lam_{n+1}x_n(\tilde v)+b_{n+1})\big)^{2} \big].
$$
\State \textbf{goto} 2.
\end{algorithmic}
\end{algorithm}

In the following proposition we prove the convergence of the output of Algorithm \ref{alg-kyle-0} to the well known Kyle equilibrium. 
\begin{proposition}\label{prop:alg-kyle-0}
Under the setting of Algorithm \ref{alg-kyle-0}, for any initial values $(\lam_0,b_0)$ we have 
$$
\lim_{n\rr \infty } (\lam_n,b_n) = (\lam^*, 0), \quad \textrm{and } \lim_{n\rr\infty} x_n = x^*(v) = \frac{v}{2\lam^*}, 
$$
where 
\be \label{lam-st-kyle} 
\lam^*=  \frac{\sigma_{\tilde v}}{2\sigma_{\tilde u}}. 
\ee
\end{proposition} 
The proof of Proposition \ref{prop:alg-kyle-0} is given in Section \ref{sec-pfs-3}. 

Next we present an Algorithm for which $\gamma>0$ in \eqref{equi-mm} and therefore the equilibrium price  has a bid-ask spread. Using this algorithm we derive $(\lambda^*, \theta^*)$ from (\ref{lam-th-eq}).

 \begin{algorithm}[H]
\caption{Equilibrium price for Kyle model with bid-ask spread}\label{alg-kyle-bid-ask}
\begin{algorithmic}[1]
\State Initialise the price function $P(x)$ with some arbitrary weights $(\lam_0,b_0,\theta_0)$.
\State
Given that $P(x) = \lam_nx+\theta_n\sign(x) + b_n$, find $x_n(v)$  that maximizes the trader's optimization problem 
\bd
 R_{v}(x)  =\mathbb E\big[ (\tilde v -(\lam_nx(\tilde v)+\theta_n\sign(x)+b_n)) x (\tilde v) | \tilde v =v  \big].
 \ed 

\State
Find $(\lam_{n+1},b_{n+1}, \theta_{n+1})$ that minimizes market maker's cost function,  
 $$
 C_{n+1}(\theta_{n+1},\lam_{n+1}) := E\big[ \big(\tilde v - (\lam_{n+1}x_n(\tilde v)+\theta_{n+1}\sign(x)+b_{n+1})\big)^{2} \big]-\gamma \theta_{n+1} E\big[|x_n(\tilde v)+\tilde u| \big].
 $$
\State \textbf{goto} 2.
\end{algorithmic}
\end{algorithm}

Now we are ready to define the notion of metastable equilibrium.
\begin{definition} [metastable equilibrium] 
We say that $(x^{*}, \lambda^{*}, \theta^{*})$ is a metastable equilibrium if for any $\alpha \in (0,1)$ there exists $\gamma (\alpha)>0$ large enough such that for every $n \geq 0$, if $(x_{n}, \lambda_{n}, \theta_{n}) = (x^{*}, \lambda^{*}, \theta^{*})$, then 
\be \label{met-stable} 
P\big((x_{n+1}, \lambda_{n+1}, \theta_{n+1}) = (x^{*}, \lambda^{*}, \theta^{*})\big) >\alpha.
\ee
 \end{definition} 

In the following Proposition we prove that there exists a metastable equilibrium for the our game and specify it. 
\begin{proposition} \label{prop-big-gamma} 
Assume the same hypothesis as in Proposition \ref{lemma-min-trader}, only now let $\tilde v - p_{0}$ be standard Gaussian and $\tilde u$ be Uniform on $[-1,1]$. 
Then, there exists a metastable equilibrium in which $X$ is linear and $P$ consist only a bid-ask spread, that is,
$$
X(v) = \frac{1}{2 \theta^{*}} (v-p_{0}),  \quad P(x) = p_{0}+ \theta^{*} sign(x), 
$$
where $\theta^{*}$ is the unique root of the function
 \begin{equation*}
 H(\theta):= -\frac{1}{\theta}  \erf(\theta \sqrt{2})- \frac{\gamma}{2} \left(  \erf(\theta \sqrt{2}) \cdot \left(1+\frac{1}{4\theta^2}\right)+ \frac{1}{\theta\sqrt{2\pi}}e^{-2\theta^2}\right)+ 2\theta .
 \end{equation*}
\end{proposition} 
The proof of Proposition \ref{prop-big-gamma}  is given in Section \ref{sec-pfs-3}.

 \section{Algorithm for finding the equilibrium}\label{section-neural}

In this section we describe the implementation of Algorithm \ref{alg-kyle-bid-ask} for finding the equilibrium. We consider the two previously used settings: standard Gaussian future price $\tilde v$ with standard Gaussian noise $\tilde u$ or the standard Gaussian future price with Uniform noise on $[-1,1]$. 

In order to solve the basic optimization problems, we use Scipy optimize python package, and for solving neural network we use the PyTorch library. In both cases, we use Python 3 on an office CPU with an i7-4930k processor and we choose a large sampling size of $N=10^5$. The average total running time of the Gaussian-Uniform algorithm and the Gaussian-Gaussian algorithm are $3.7s$ and $18.3 s$, respectively. In terms of complexity, the complexity of the Gaussian-Gaussian case is much higher as we do not have a closed form formula for the insider's optimiser. Thus, for each price value $v$, we need to solve an additional optimization problem.

\subsection{Designing a neural network to find the equilibrium}

In step 3 of Algorithm \ref{alg-kyle-bid-ask} we need to find the parameters of the price function $P$ in order to obtain the minimum of a the Market Maker's cost function $C$; this is typically what neural networks are doing.
Here is a minimalist description of this class of approximators, for additional information the reader is referred to \cite{vapnik} and \cite{Goodfellow}.


Neural networks are parametrized functions, mapping a $K$ dimensional vector of \emph{inputs} $X$ to a vector of \emph{outputs} $Y$. In our case the output will be a vector of length $2$, namely $(\lam^*,\theta^*)$. 
In order to produce the output, inputs are first mapped to a \emph{hidden layer of $h$ neurons}, by combining linearly the inputs via \emph{weights} $(w_{i,k})_{i,k}$ and a biases $(b_i)_i$, and then applying an \emph{activation function} $\phi$ to this combination:
$X\mapsto (\phi(\sum_{k=1}^K w_{i,k} X_k + b_i))_{1\leq i\leq h}$. This operation is repeated several times. Ultimately the last hidden layer is mapped to the output in a similar way.

One of the main features of neural networks is that the weights and biases of each layer can be \emph{trained} to minimize a \emph{loss function}, thanks to automatic differentiation methods (see \cite{geeraert2017mini} for detailed applications of Adjoint
Algorithmic Differentiation in finance).
Once a loss function is specified, the theory of statistical learning studies how minimizing the expectation of the loss function over a distribution can be performed on a sample of this distribution (see for instance \cite{vayatis1999distribution} or more recently \cite{choromanska2015loss} for related work on deep neural networks). One typical \emph{empirical loss function} is the 
well known $L^{2}$ loss function:
\be \label{loss}
    Loss = \frac{1}{N} \sum_{j=1}^N\big\|f_{W,b}(x_j)- y_j \big\|^2 \xrightarrow{N\rightarrow \infty} \mathbb{E}\|f_{W,b}(\bold{x})- \bold{y} \big\|^2, 
\ee
where $\bold{x}$ and $\bold{y}$ are random variables with the same law as the inputs and the outputs, respectively. Clearly convergence to the expectation takes place only under certain assumptions on the distributions of the inputs and the outputs.

We will encode step 3 of Algorithm \ref{alg-kyle-bid-ask} in  architecture of a neural network in order to leverage on their learning capabilities (via automatic differentiation). Note that in our case the distribution of the datasets are known, hence we can generate very large samples using Monte-Carlo methods. The convergence of the empirical loss function to the theoretical one in \eqref{loss} is guaranteed. Moreover, note
that step 2 in this Algorithm \ref{alg-kyle-bid-ask} can be solved efficiently by means of Propositions \ref{gaus} and \ref{prop-unif}. 

The iterations between step 2 and step 3 of Algorithm \ref{alg-kyle-bid-ask} can be seen as an \emph{adversarial approach}, using the language of the machine learning community (see \cite{cao2020connecting} for connections between adversarial learning methods and mean field games). In Stackelberg games one player computes her optimal control for multiple different scenarios, and then the other player chooses the scenario that is the best for her (see \cite{moon2018linear} for details). In our Kyle game, the insider computes his optimal response for any value of $\lambda$ and $\theta$, and then the market maker chooses  $\lambda$ and $\theta$ that minimizes her costs.
It is straightforward that the market maker's choice is adversarial to the informed trader, and it will be implemented via a neural network. We are thus iterating sequences of (1) learning of the neural network, (2) adversarial choice by the insider, up to convergence. This is compatible with the definition of \emph{adversarial learning}.

Figure \ref{fig:twoLayer} describes the two layers network architecture which is used to solve the optimisation problem in step 3 of Algorithm \ref{alg-kyle-bid-ask}. In the first layer we have two neurons, one which multiplies the order flow by $\lambda$ and adds a bias parameter $b$. The other neuron receives the order flow input and applies the sign activation function to it. In the second layer, the output of the bottom neuron is multiplied by the parameter $\theta$ and combined with the output of the top neuron. Overall the output takes the form $y=\lambda x +\theta sign(x)+b$ which is compatible with \eqref{admis-p}.

\begin{figure}[h!]
    \centering
    \includegraphics[scale=0.38]{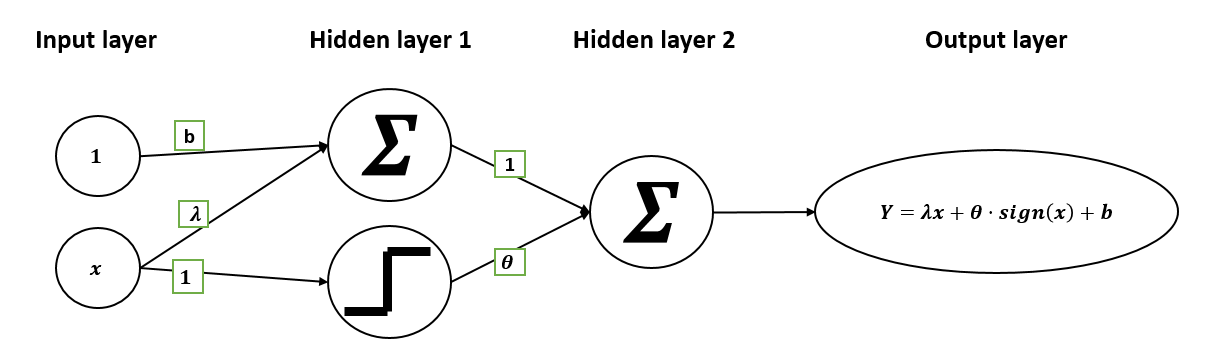}
    \caption{Two layers network for solving the market maker's optimization problem}
    \label{fig:twoLayer}
\end{figure}

Based on this neural network, we present an algorithm which derives the equilibrium price via finding $(\lambda^{*},\theta^{*})$ and the optimal market order of the insider $x^{*}(v)$. 

 \begin{algorithm}[H] \label{nn-alg}
\caption{Two layers network for the market maker price with bid-ask spread}\label{alg-kyleBidAsk-NN}
\begin{algorithmic}[1]
\State Initialise the price function $P(x)$ with some arbitrary weights $(\lam_0,b_0,\theta_0)$. 
\State
Sample $v_1,...,v_N$ i.i.d distributes according to the law of $\tilde v$.
\State
Find $x_n(v_{i})$, $i=1,...,N$ that maximize the trader's optimisation problem: 
\bd
 R_{v}(x)  =\mathbb E\big[ (\tilde v -(\lam_n x_n(\tilde v)+b_n)+\theta_n sign(x_n) )\tilde x_n(\tilde v) | \tilde v =v_{i}  \big].
\ed

\State
Sample a set of $u_1,...,u_{N}$ i.i.d distributes according to the law of $\tilde u$.
\State
Train the neural network presented in Figure \ref{fig:twoLayer} with inputs $\{(u_1+x_n(v_{1}), v_1),...,(u_{N}+x_n(v_{N}),v_N)\}$ and $(\lam_{n},b_{n},\theta_n)$ as the initial weights.
\State
Extract the weights $(\lam_{n+1},b_{n+1},\theta_n)$ that minimize the loss function
\bd
\frac{1}{N} \sum_{j=1}^N\big(P(x_n(v_{j}) +u_j) - v_j \big)^2 -\gamma \theta\frac{1}{N}\sum_{j=1}^N |x_n(v_{j}) +u_j|. 
\ed
\State
Update the price function $P$ according to the new weights $$ P(x)=\lam_{n+1}x+\theta_n sign(x) + b_{n+1} $$

\State \textbf{goto} 2.
\end{algorithmic}
\end{algorithm}

In figure \ref{conv-alg} we illustrate the convergence of our neural network algorithm for $\gamma=0.5$. We plot $(\lambda_{n},\theta_{n})$ as a function of $n$. We observe that the convergence of the algorithm is very fast, as we achieve convergence to equilibrium in accuracy of $10^{-6}$ after only $14$ iterations in the Gaussian noise case and after $9$ iterations in the Uniform noise case. The number of required iterations depends also on the risk-aversion parameter. Our results suggest that when the risk-aversion parameter is close to zero it takes an average of $9$ iterations to converge and when the risk-aversion parameter is close to $1$ it takes an average of $13$ iterations to converge.

\begin{figure}[h!]
  \begin{subfigure}[b]{0.5\textwidth}
    \includegraphics[width=1.1\textwidth]{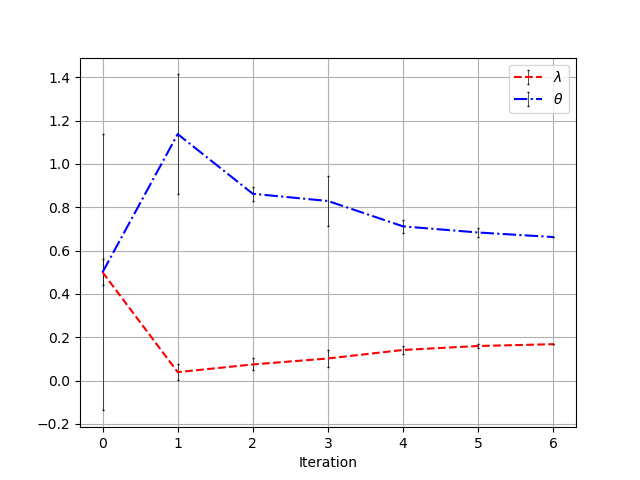}
    \caption{Normal Noise Case}
     
  \end{subfigure}
  \hfill
  \begin{subfigure}[b]{0.5\textwidth}
    \includegraphics[width=1.1\textwidth]{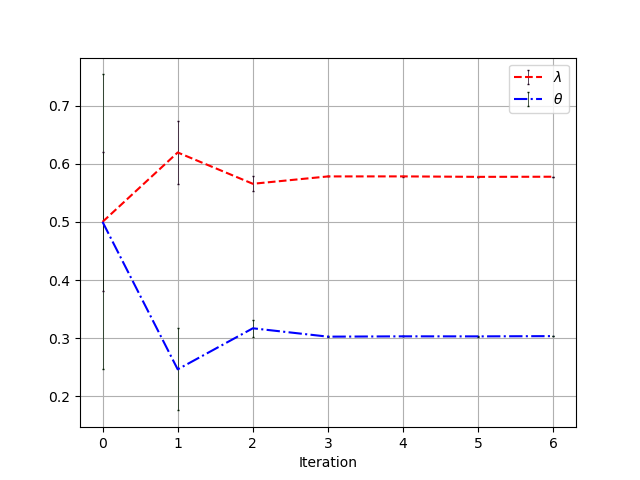}
    \caption{Uniform Noise Case}
     
  \end{subfigure}
  \caption{Plot of $\lambda_{n}$ and $\theta_{n}$ from Algorithm \ref{alg-kyleBidAsk-NN} as a function of $n$--the number of iterations (i.e. the numbers of times running from step 2 to step 8). The risk-aversion parameter is $\gamma=0.5$.} \label{conv-alg}
\end{figure}



\subsection{Enlarging the class admissible prices}  \label{sec-en-admis}

A possible generalisation of  Algorithm \ref{alg-kyleBidAsk-NN} would be to derive the price $P$ from an arbitrary neural network as described in Algorithm \ref{alg-kyle-3}. This will allow us to depart from the class of admissible prices $\mathscr{P}$ in \eqref{admis-p}. However, if the price function $P$ is general, deriving the insider's strategy becomes more involved as the results of of Propositions \ref{gaus} and \ref{prop-unif} do not apply. As a result the implementation of such algorithm requires to solve numerically at each steps a large number of optimisation problems. This creates additional discretization errors that are detrimental for the convergence to equilibrium.

  \begin{algorithm}[H]
\caption{Multi-layer network for the market maker price}\label{alg-kyle-3}
\begin{algorithmic}[1]
\State Initialise the price function $P(x)$ with a random seed.  
\State
Sample $v_1,...,v_N$ i.i.d distributed according to the law of $\tilde v$ and independently sample a matrix $U=(u_{i,j})_{i=1,...,M, \, j=1,...N}$ where $u_{i,j}$ are i.i.d distributed according to the law of $\tilde u$.
\State
Solve the optimization problem 
\bd
\min_{x_n(v_{j})} \frac{1}{M} \sum_{i=1}^M\big(P(x_{n}(v_{j})+u_{i,j}) - v_j\big)x_n(v_{j}), \quad \textrm{for any } j=1,...,N.
\ed 
\State
Sample a new set of $u_1,...,u_N$ i.i.d distributed according to the law of $\tilde u$.
\State
Train the neural network with 
${\{(u_1+x_n(v_{1}), v_1),...,}{(u_N+x_n(v_{N}),v_N)\}}$ as inputs, using as initial weights the ones obtained at the previous iteration.
For the training, use the following loss function
\bd
\frac{1}{N} \sum_{j=1}^N\big(P(x_{n}(v_{j}) +u_j) - v_j \big)^2 -\gamma \theta\frac{1}{N}\sum_{j=1}^N |x_n(v_{j}) +u_j|. 
\ed
\State
Update the price function $P$ according to weights extracted in the previous step 

\State \textbf{goto} 2. 
\end{algorithmic}
\end{algorithm}

In order to enlarge the class of admissible price functions while preserving the complexity of Algorithm \ref{alg-kyleBidAsk-NN} we tested this algorithm with higher degree polynomials in addition to the sign function. More precisely, we tested Algorithm \ref{alg-kyleBidAsk-NN} with the following classes of price functions
\begin{align*}
    &\mathscr{P}_3 (x) = \{P(x)=\lambda_1 x + \lambda_3 x^3 + \theta \textrm{sign}(x)+p_{0}, \, \lambda_1,\lambda_3, \theta > 0\},
    \\&
    \mathscr{P}_5 (x) = \{P(x)=\lambda_1 x + \lambda_3 x^3 +\lambda_5 x^5+ \theta \textrm{sign}(x)+p_{0}, \, \lambda_1,\lambda_3,\lambda_5, \theta > 0\},
    \\&
    \mathscr{P}_7 (x) = \{P(x)=\lambda_1 x + \lambda_3 x^3 +\lambda_5 x^5+ +\lambda_7 x^7+ \theta \textrm{sign}(x)+p_{0}, \, \lambda_1,\lambda_3,\lambda_5,\lambda_7 \theta > 0\}.
\end{align*}

In all cases, the weights, besides $\lambda_1$ and $\theta$ converged to zero. This leads us to the following conjecture.
\begin{conjecture}
Let $\mathscr{P}_{n}$ be the class of polynomial price functions of order $n$ which incorporates a bid-ask spread. That is,
$$
{\mathscr{P}}_{n} = \{P(x)= Poly_{n}(x)+ \theta sign(x) \}.
$$
Then, there exists an equilibrium with all coefficients equal to zero except for a linear coefficient and the spread $\theta$, for $\gamma \in (\gamma_{LBid}, \gamma_{Bid})$. 
\end{conjecture}

\section{Proofs of Proportions \ref{lemma-min-trader},  \ref{gaus} and \ref{prop-unif}} \label{sec-prf} 
This section is dedicated to the proofs of Proportions \ref{lemma-min-trader},  \ref{gaus} and \ref{prop-unif}. Before we start with the proofs, we introduce some notation and auxiliary lemmas. 

\begin{proof}[Proof of Proposition \ref{lemma-min-trader}]
Let $P \in \mathscr P$, and fix $v \in \re$. From (\ref{admis-p}) and since $\tilde p =P(\tilde x +\tilde u)$ we have 
\be \label{r1} 
\begin{aligned} 
R_{v}(x) &:=\mathbb E\big[ \tilde \pi(X,P) | \tilde v =v  \big] \\
&= \E\big[ \big(\tilde v -\lambda (x+\tilde u)-\theta \textrm{sign}( x +\tilde u)-p_{0}\big)x| \tilde v =v  \big] \\
&= \big( v -\lambda x-\theta \E[ \textrm{sign}( x +\tilde u)]-p_{0}\big)x,
\end{aligned} 
\ee
where $\lam, \theta >0$. We have used the fact that $E(\tilde u)=0$ and the independence between $\tilde v$ and $\tilde u$.  
Recall that $F_{\tilde u}$ is the cumulative distribution function of $\tilde u$. Since $\tilde u$ is symmetric we have 
$$
\E[ \textrm{sign}( x +\tilde u)] = 1-2F_{\tilde u}(-x).
$$  
Therefore (\ref{r1}) becomes   
\be\label{r2} 
\begin{aligned}
R_{v}(x) &= \big( v -\lambda x-\theta(1-2F_{\tilde u}(-x))-p_{0}\big)x \\ 
&=-\lam x^2+(v-p_0-\theta)x +2\theta xF_{\tilde u}(-x). 
\end{aligned} 
\ee
Note that $R_v(0) =0$. Now assume that $v-p_{0} >0$. Since $F_{\tilde u}(0) = 1/2$, it is easy to verify that $R'_v(0)>0$ and therefore there exists $x>0$ such that $R_v(x)>0$. 
Moreover note that $\lim_{x\rr \pm \infty } R_{v}(x) = -\infty$, and that for any $x>0$ we have $R_v(x) >R_v(-x)$.  It follows that there exists $0<x^*(v)<\infty $ that maximizes $R_{v}$. Moreover $x^{*}(v)$ is a solution to the equation $R'_{v}(x)=0$, which is equivalent to \eqref{deriv-cond}.

For the case $v-p_{0}=0$ we have that $R_{v}(x)<0$ when $x\not =0$, and therefore $x^*(v) = 0$ is the unique maximizer of $R_v$.  

In the case when $v-p_0 <0$ we have $R'_v(0)<0$, and by repeating the same steps as in the case $v-p_0 >0$, it follows that there exists $\infty<x^*(v)<0 $ that maximizes $R_{v}$. Moreover $x^{*}(v)$ is a solution to the equation $R'_{v}(x)=0$.

We conclude that when $v-p_{0}\not = 0$ there exists a unique maximum to (\ref{r1}) on $(-\infty, \infty)$ which we denote by $x^{*}=x^{*}(v)$. We also have that $R_{v}(x^{*}(v))>0$ and $\sign(x^*(v)) = \sign(v-p_0)$ . Moreover, when $v-p_{0}=0$, the unique maximum to (\ref{r1}) is $x^{*}=0$, for which we have $R_{0}(0)=0$.
\end{proof} 
 
\begin{proof}[Proof of Proposition  \ref{gaus}] 
Again we assume that $v-p_0>0$, where the case $v-p_0<0$ can be handled similarly. The existence of a unique maximizer $x^*(v)$ to $R_v$ is known from Proposition \ref{lemma-min-trader}. It is also known that $x^*(v)$ satisfies $R'_v(x^*(v))=0$, which in this case is given by \eqref{normal-eqn}. Hence it remains to identify the zeros of $R'_v$. 

Recall that $\tilde u$ is a mean-zero Gaussian with variance $\sigma^2_{\tilde u}$. From \eqref{r2} it follows that the second derivative of the $R_{v}$ is given by  
\begin{equation*}
 R''_{v}(x) = -2\lam -4\theta f_{\tilde{u}}(-x)+ 2\theta x f_{\tilde{u}}^{'}(-x).
\end{equation*}
Without loss of generality, we assume that $\sigma^2_{\tilde u}=1$, so we have
\begin{equation*}
R''_{v}(x) = -2\lam -2\theta \frac{1}{\sqrt{2\pi}}e^{-x^2/2}\left (2-x^{2} \right ).
\end{equation*}
It can be easily verified that $R''_{v}(x)$ is monotone increasing on $[0,2)$ and then decreasing on $[2,\infty)$. Since $R''_{v}(0)<0$ we get that the derivative of $ R'_{v}(x)$ satisfies one of the two cases: (i) negative on $[0,\infty)$, (ii) alternate signs twice on $[0,\infty)$, negative-->positive-->negative. Combining this with the fact that $ R'_{v}(0)>0$ and $\lim_{x\rr  \infty } R'_{v}(x) = -\infty$, we get that in case (i) there is only one solution to $R'_{v}(x)=0$ which is the global maxima. 

In case (ii) there are either one or three solutions ($0<x_1<x_2<x_3$) to $R'_{v}(x)=0$. The one solution case is clearly a global maxima as in case (i). If in case (ii) there are three solutions, then either $x_1$ or $x_3$ must be the unique global maxima.   

\end{proof}


\begin{proof} [Proof of Proposition \ref{prop-unif} ]
By proposition \ref{lemma-min-trader}, if $v-p_0>0$ then $x>0$. We also note that when $\tilde u$ is distributed uniformly on $[-1,1]$ and \eqref{r1} is given by 
\be \label{r-unif} 
  R_{v}(x) =
  \begin{cases}
     -(\lam+\theta)x^2 +(v-p_0)x,      & \text{for } 0\leq x \leq 1, \\
    -\lam x^2+(v-p_0-\theta)x, & \text{for } x>1.
  \end{cases}
\ee
Define $R_{1}(x) =-(\lam+\theta)x^2 +(v-p_0)x$ and $R_{2}(x)=-\lam x^2+(v-p_0-\theta)x$, where we note that both $R_1$ and $R_2$ are convex parabolas. The maxima $x_i$ of $R_i$ , $i=1,2,$ are obtained at
\be \label{x-i}
x_1 = \frac{v-p_0}{2(\lam+\theta)},  \quad x_2 = \frac{v-p_0-\theta}{2\lam}, 
\ee
and we have 
\be\label{r-i}
R_1(x_1) = \frac{(v-p_0)^2}{4(\lam+\theta)},  \quad R_2(x_2) = \frac{(v-p_0-\theta)^2}{4\lam}. 
\ee
Note that the points $(x_i, R_{i}(x_{i}))$ $i=1,2$ appear both in the graph of $R_{v}$ if $0\leq x_1 \leq 1$ and $x_2\geq 1$ which translates to 
$v-p_0 \leq 2(\lam+\theta)$ and $v-p_0 \geq 2\lam +\theta$ (respectively). It follows in order to find the global maxima $x^*$ when $2\lam +\theta \leq v-p_0 \leq 2(\lam+\theta)$ we need to compare $R_1(x_1)$ to $R_2(x_2)$. 

We get that $R(x_1) > R(x_2)$ , i.e. $x^*=x_1$, when $2\lam +\theta  \leq v-p_0 \leq \bar z$ where $\bar z = \lam+\theta + \sqrt{(\lam+\theta)\lam}$. Moreover, when $\bar z \leq v-p_0\leq 2(\lam+\theta)$, then $R(x_1) > R(x_2)$ which means that $x^*=x_2$ . In order to complete the proof we need to show that  
\be \label{unif-verf} 
\begin{aligned} 
x^*=x_1, & \quad \textrm{when }  0< v-p_0 \leq2\lam +\theta,  \\
x^*=x_2, & \quad \textrm{when }  2(\lam+\theta)< v-p_0.  
\end{aligned} 
\ee
Note that if $0< v-p_0 \leq2\lam +\theta$ then $(x_1,R_1(x_1))$ appears in the graph of $R_v$ but $(x_2,R_2)$ doesn't. From (\ref{r-unif}) it follows that $R_v$ is decreasing for $x\geq 1$. Since $x_1$ is the maximum of $R_v$ on $[0,1]$ if follows that $x^*=x_1$. 

If $2(\lam+\theta)< v-p_0$, then $(x_2,R_2(x_2))$ appears in the graph of $R_v$ but $(x_1,R_1)$ doesn't. From (\ref{r-unif}) it follows that $R_v$ is increasing on $[0,1]$. Since $x_2\geq 1$ is the maxima of $R_2$, it is also the maxima of $R_v$ and $x^*=x_2$, and we verify (\ref{unif-verf}). 
\end{proof} 

\section{Proofs of Theorems \ref{thm-equil1} and Proposition \ref{prop-lin-eq}} \label{sec-pfs2} 
\begin{proof}[Proof of Theorem \ref{thm-equil1}]
For $x^*(v)$ as in Proposition \ref{lemma-min-trader} define 
\be \label{c-lam}
C(\theta,\lam)=  \E\big[ \big( \tilde v- p_{0} - \lambda (x^{*}(\tilde v)+\tilde u) - \theta \textrm{sign}(x^{*}(\tilde v)+\tilde u) \big)^{2}\big]- \gamma\theta E[|x(\tilde v) + \tilde u|] . 
\ee
From (\ref{admis-p}) and (\ref{equi-mm}) it follows that we need to solve the minimization problem 
\bd
\min_{(\theta,\lam) \in \re^{2}_{+}} C(\theta,\lam),  
\ed
where $\re^{2}_{+}$ denotes the first quadrant of $\re^{2}$.

Using the independence of $\tilde v$ and $\tilde u$ we get the following first order conditions: 
 \be \label{partial-c} 
\begin{aligned} 
\partial_{\lam} C(\theta,\lam)& = -2\E\big[(x^{*}(\tilde v)+\tilde u) \big( \tilde v- p_{0} - \lambda (x^{*}(\tilde v)+\tilde u) - \theta \textrm{sign}(x^{*}(\tilde v)+\tilde u) \big)\big] \\
 &= -2\E\big[x^{*}(\tilde v)( \tilde v- p_{0}) \big]+2\lam\E\big[(x^{*}(\tilde v)+\tilde u)^{2}\big] +2\theta \E\big[|x^{*}(\tilde v)+\tilde u|\big]   \\
 &=0, 
\end{aligned} 
\ee
and  
\be \label{partial-c2} 
\begin{aligned} 
\partial_{\theta} C(\theta,\lam) =& -2\E\big[ \textrm{sign}(x^{*}(\tilde v)+\tilde u) \big( \tilde v- p_{0} - \lambda (x^{*}(\tilde v)+\tilde u) - \theta \textrm{sign}(x^{*}(\tilde v)+\tilde u) \big)\big]  \\ 
& - \gamma E[|x^{*}(\tilde v) + \tilde u|]  \\
 =&  -2\E\big[ \textrm{sign}(x^{*}(\tilde v)+\tilde u)( \tilde v- p_{0})\big]+(2\lambda-\gamma) \E\big[|x^{*}(\tilde v)+\tilde u |\big]+ 2\theta   \\
 =&0. 
\end{aligned} 
\ee
We arrive to the following linear system of equations:   
\bd
\begin{pmatrix}
    \ell_{2}       & \ell_{1}  \\
    \ell_{1}       & 1  \\
\end{pmatrix}
\cdot 
\begin{pmatrix}
    \lam       \\
    \theta      \\
\end{pmatrix}
=
\begin{pmatrix}
    \mu\\
    \kappa +\gamma \ell_{1}/2     \\
\end{pmatrix}.
\ed
It follows that 
\be \label{opt-eqn}
\lam^{*}  = \frac{\mu -(\kappa +\gamma \ell_{1}/2) \ell_{1}}{\ell_{2}-\ell_{1}^{2}}, \quad \theta^{*} = \kappa+\frac{\gamma}{2}\ell_{1}- \frac{\ell_{1}(\mu -\kappa-\ell_{1}\gamma/2)}{\ell_{2}-\ell_{1}^{2}}.
\ee
The Hessian matrix of $C(\theta,\lam)$, 
$$
H(C(\theta,\lam)) = \begin{pmatrix}
    2\ell_{2}       & 2\ell_{1}  \\
    2\ell_{1}       & 2 
\end{pmatrix}, 
$$
is positive definite, hence $(\lam^{*},\theta^{*})$ is a global minima.
\end{proof} 

  
Before we prove Proposition \ref{prop-lin-eq} we introduce the following lemma. 
\begin{lemma} \label{lemma-z-u} 
Let $Y$ and $Z$ be independent random variables, such that $Z$ is a centred Gaussian with variance $\sigma_{\tilde v}^2$.  Assume further that one of the following assumptions holds: 
\begin{itemize} 
\item[\textbf(a)]  $Y$ is a centred Gaussian with variance $\sigma^2_{\tilde v}$. 
\item[\textbf(b)] $Y$ is a Uniform random variable on $[-b,b]$, for some $b\ge\sqrt{3}\sigma_{\tilde v}$. 
 \end{itemize} 
 Then we have 
\be \label{z-y-lemma} 
 \E\big[|Z+ Y |\big]  \geq 2 \E\big[ \rm{sign}(Z+ Y)Z\big]. 
\ee
\end{lemma}

\begin{proof} 
\textbf{(a)} Note that 
\bn
 \E[|Z+ Y |] = \E[ (Z+Y) \mathds{1}_{\{Z+Y>0\}}] - \E[ (Z+Y) \mathds{1}_{\{Z+Y \leq 0\}}].
\en
On the other hand, 
\bn
 \E\big[ \textrm{sign}(Z+ Y)Z\big] = \E[ Z \mathds{1}_{\{Z+Y>0\}}] - \E[ Z \mathds{1}_{\{Z+Y \leq 0\}}]. 
\en
Hence in order to prove (\ref{z-y-lemma}) we need to show that 
\be \label{xy1} 
\E[ Y \mathds{1}_{\{Z+Y>0\}}] - \E[ Y \mathds{1}_{\{Z+Y \leq 0\}}] \geq \E[ Z \mathds{1}_{\{Z+Y>0\}}] - \E[ Z \mathds{1}_{\{Z+Y \leq 0\}}],
\ee
or 
\be \label{xy1} 
 \E\big[ \rm{sign}(Z+ Y)Y\big] \geq   \E\big[ \rm{sign}(Z+ Y)Z\big]. 
\ee
Since $Z$ and $Y$ have the same law, then the inequality above holds trivially in an equality. 

\textbf{(b)} Note that 
\bd \begin{aligned} 
&\E\big[|Z+ Y |\big]  - 2 \E\big[ \rm{sign}(Z+ Y)Z\big] \\
&=
\int \int_{z+y>0}(z+y)f_{Z}(z)f_{Y}(y)dzdy-\int \int_{z+y\leq0}(z+y)f_{Z}(z)f_{Y}(y)dzdy
\\
&
\quad - 2 \int \int_{z+y>0}zf_{Z}(z)f_{Y}(y)dzdy+2\int\int_{z+y \leq 0}zf_{Z}(z)f_{Y}(y)dzdy, 
 \end{aligned} 
\ed
where $f_{Y}$ and $f_{Z}$ are the probability densities of $Y$ and $Z$. 

Since $Y$ is uniformly distributed on $[-b,b]$ and $Z$ is a centred Gaussian with variance $\sigma_{\tilde v}$ we get
\bn
&&\E\big[|Z+ Y |\big]  - 2 \E\big[ \rm{sign}(Z+ Y)Z\big] \\
&&=
\frac{1}{2b}\int_{-\infty}^{\infty}\int_{-b}^{b}1_{\{z+y>0\}}(y-z) f_{Z}(z)dzdy+\frac{1}{2b}\int_{-\infty}^{\infty}\int_{-b}^{b}1_{\{z+y>0\}}(z-y)f_{Z}(z)dzdy \\ 
&&= \frac{1}{b}\int_{-b}^{b}\int_{-y}^{\infty}(y-z) f_{Z}(z)dzdy \\ 
&&= \frac{1}{\sqrt{2\pi} \sigma_{\tilde v}}  \frac{1}{b}\int_{-b}^{b}\int_{-y}^{\infty}(y-z) e^{-z^{2}/(2\sigma^{2}_{\tilde v})}dzdy. 
\en

Calculation of the above integral gives: 
\bn
&&\E\big[|Z+ Y |\big]  - 2 \E\big[ \rm{sign}(Z+ Y)Z\big] \\
&&=  \frac{1}{b}\int_{-b}^{b} \left(\frac{1}{2}y\left( \erf\left(\frac{y}{\sqrt{2}\sigma_{\tilde v}}\right)+1\right) -  \frac{\sigma_{\tilde v}}{\sqrt{2 \pi}}  e^{-y^{2}/(2\sigma^{2}_{\tilde v})}  \right)dy  \\
&&=   \frac{1}{2b} (b^{2}-\sigma_{\tilde v}^{2})\erf\left(\frac{b}{\sqrt{2}\sigma_{\tilde v}}\right)+  \frac{\sigma_{\tilde v}}{ \sqrt{2\pi}} e^{-b^{2}/(2\sigma_{\tilde v }^{2})}  - \frac{\sigma_{\tilde v}^{2}}{b} \erf\left(\frac{b}{\sqrt{2}\sigma_{\tilde v}}\right) \\
&&= \frac{1}{2b}\left(b^{2} -3 \sigma_{\tilde v}^{2}  \right)
\erf\left(\frac{b}{\sqrt{2}\sigma_{\tilde v}}\right)  +  \frac{\sigma_{\tilde v}}{ \sqrt{2\pi}} e^{-b^{2}/(2\sigma_{\tilde v }^{2})}.
\en
Thus, when $b\ge\sqrt{3}\sigma_{\tilde v}$, \eqref{xy1} holds and the result follows.
\end{proof}

\begin{proof}[Proof of Proposition \ref{prop-lin-eq}]
Recall that $C(\theta,\lam)$ was introduced in (\ref{c-lam}). In order to prove Proposition \ref{prop-lin-eq} we show that $(\theta^{*}, \lam^{*})= (0,\frac{\sigma_{\tilde v}}{2\sigma_{ \tilde u}})$ minimizes $C(\theta,\lam)$ for $\gamma =0$. 
In the proof Theorem \ref{thm-equil1} we showed that $C(\theta,\lam)$ is concave. Therefore, it is enough to show that 
\be
\partial_{\lam} C(0,\lam^{*})  =0, \quad  \textrm{and }\partial_{\theta} C(0,\lam^{*})  \geq 0.
\ee
Recall that the traders optimal order size is $x^*(v)=\frac{v-p_0}{2\lam}$. From (\ref{c-lam}) we get 
\bd \begin{aligned} 
\partial_{\lam} C(0,\lam) &= -2\E\big[x^{*}(\tilde v)( \tilde v- p_{0}) \big]+2\lam\E\big[(x^{*}(\tilde v)+\tilde u)^{2}\big] \\
 &=-\frac{1}{\lam}\E\big[( \tilde v- p_{0})^{2} \big]+2\lam\E\Big[\Big(\frac{\tilde v-p_{0}}{2\lam}+\tilde u\Big)^{2}\Big] \\
 &=-\frac{1}{2\lam}\E\big[( \tilde v- p_{0})^{2} \big]+2\lam E(\tilde u^{2}). 
\end{aligned} 
\ed
We therefore get that $\partial_{\lam} C(0,\lam^{*}) =0$. 
From (\ref{c-lam}) we also have  
\bd
\begin{aligned} 
\partial_{\theta} C(0,\lam)  =&  -2\E\big[ \textrm{sign}(x^{*}(\tilde v)+\tilde u)( \tilde v- p_{0})\big]+2\lambda \E\big[|x^{*}(\tilde v)+\tilde u |\big]  \\
=&  -2\E\Big[ \textrm{sign}\Big(\frac{\tilde v-p_{0}}{2\lam}+\tilde u\Big)( \tilde v- p_{0})\Big]+2\lambda \E\Big[\Big|\frac{\tilde v-p_{0}}{2\lam}+\tilde u \Big|\Big]. 
\end{aligned} 
\ed
Hence in order to prove that $\partial_{\theta} C(0,\lam^{*})  \geq 0$ we need to show 
\be \label{rf1} 
 \E\big[|\tilde v-p_{0}+2\lam^* \tilde u |\big]  \geq 2 \E\big[ \textrm{sign}(\tilde v-p_{0}+2\lam^* \tilde u)( \tilde v- p_{0})\big].
\ee
Since $\tilde v-p_{0}$ is a centred Gaussian with variance $\sigma^{2}_{\tilde v}$, \eqref{rf1}  follows immediately from Lemma \ref{lemma-z-u}(a) for the case where $\tilde u$ is a centred Gaussian with variance $\sigma^{2}_{\tilde u}$ and $2\lam^* \tilde u=\frac{\sigma_{\tilde v}}{\sigma_{ \tilde u}} \tilde u$. When $\tilde u$ is a centred Uniform with variance $\sigma_{\tilde u}^{2}$, then $2\lam^* \tilde u = \frac{\sigma_{\tilde v}}{\sigma_{ \tilde u}} \tilde u$ is distributed uniformly on $[-\sqrt{3} \sigma_{\tilde v} , \sqrt{3} \sigma_{\tilde v}]$ and \eqref{rf1} follows from Lemma \ref{lemma-z-u}(b).
\end{proof}

\section{Proofs of Propositions \ref{prop:alg-kyle-0} and \ref{prop-big-gamma}} \label{sec-pfs-3}


\begin{proof} [Proof of Proposition \ref{prop:alg-kyle-0}] 
Without loss of generality we assume that $p_0=0$. Let $\lam_0>0$ and solve the trader's problem from step (ii) of Algorithm \ref{alg-kyle-0} by using (\ref{x-no-sprd}) to get 
$$
x_0(v) = \beta_0v, 
$$ 
where $\beta_0= \frac{1}{2\lam_0}$. 

Now solve the optimization on step (iii) (see e.g equation (2.8) in \cite{kyle85}) to get, 
\bn
\lam_1 &=& \frac{\beta_0 \sigma_{\tilde v}^2}{\beta^2_0\sigma_{\tilde v}^2+\sigma_{\tilde u}^2} \\ 
&=& \frac{2\lam_0 \sigma_{\tilde v}^2}{ \sigma_{\tilde v}^2+4\lam_0^2\sigma_{\tilde u}^2}.
\en
Repeating this procedure of $n$ steps, we have 
\be \label{lam-n} 
\lam_n =  \frac{2\lam_{n-1} \sigma_{\tilde v}^2}{ \sigma_{\tilde v}^2+4\lam_{n-1}^2\sigma_{\tilde u}^2}, \quad \textrm{and } x_n(v) = \frac{v}{2\lam_n}.
\ee
We show that $\{\lam_n\}_{n\geq 0}$ is a non-increasing sequence if $\lam_0>\lam^*$. This can be verified by induction. The claim for $n=0$ is satisfied by the hypothesis. Assume that $\lam_{n-1} \geq \lam^*$, then we get from (\ref{lam-n}) and (\ref{lam-st-kyle}) that 
\bn
\frac{\lam_{n}}{\lam_{n-1}} &=&\frac{2  \sigma_{\tilde v}^2}{ \sigma_{\tilde v}^2+4\lam_{n-1}^2\sigma_{\tilde u}^2} \\
&\leq & \frac{2  \sigma_{\tilde v}^2}{ \sigma_{\tilde v}^2+4(\lam^*)^2\sigma_{\tilde u}^2} \\
&= &1. 
\en 
Hence $\{\lam_n\}_{n\geq 0}$ is a non-increasing sequence. In the same way we can show that if $\lam^*\geq \lam_0$, then $\{\lam_n\}_{n\geq 0}$ is a non-decreasing sequence. From these two claims together, it follows that if $\lam_0>\lam^*$ (or $\lam^*\geq \lam_0$) then $\lam^*$ is a lower bound (respectively an upper bound) of $\{\lam_n\}_{n\geq 0}$. Hence in each of these cases the limit $\lam_\infty = \lim_{n\rr \infty}\lam_n$ exists. 
From (\ref{lam-n}) it follows that 
\bd
\lam_\infty =  \frac{2\lam_{\infty} \sigma_{\tilde v}^2}{ \sigma_{\tilde v}^2+4\lam_{\infty}^2\sigma_{\tilde u}^2}, 
\ed
therefore $\lam_{\infty} = \lam^*=  \frac{\sigma_{\tilde v}}{2\sigma_{\tilde u}}$. We also have $\lim_{n\rr\infty} x_n(v) = \frac{v}{2\lam_\infty} = \frac{v}{2\lam^*}$. 
\end{proof} 

\begin{proof}[Proof of Proposition \ref{prop-big-gamma}] 
Without loss of generality we assume that $p_0=0$. Let $\gamma >0$. Assume that $\lambda=0$ and $x(\tilde v)=\frac{\tilde{v}}{2\theta}$ for some $\theta>0$ to be determined.  

We will first show that there exists $\theta^{*}>0$ such that 
\be \label{der-cond} 
\partial_{\lam} C(\theta^{*},0) > 0, \qquad \partial_{\theta} C(\theta^{*},0) =0. 
\ee
Note that from \eqref{partial-c} we have in this case, 
\begin{equation}\label{partial-c_uniform}
\partial_{\lam} C(\theta,0)
 =
-2\E\left[\frac{\tilde{v}}{2\theta}\tilde{v}\right]+ 2\theta \E\left[\left|\frac{\tilde{v}}{2\theta}+\tilde u\right|\right]
=
-\frac{1}{\theta}+ 2\theta \E\left[\left|\frac{\tilde{v}}{2\theta}+\tilde u\right|\right]. 
\end{equation}

Recall that $\tilde v$ is a standard Gaussian. Using the tower property we have, 
\bn
    \E\left[\left|\frac{\tilde{v}}{2\theta}+\tilde u\right|\right]
    &=& \E\left[   \E\left[\left|\frac{\tilde{v}}{2\theta}+\tilde u\right| \, \bigg|\tilde u\right]  \right] \\
    &=&\E\left[\frac{1}{\sqrt{2\pi}\theta}e^{-2\tilde{u}^2 \theta^2}+ \tilde{u}\cdot \erf(\tilde{u}\theta \sqrt{2})
    \right].
\en
Since $\tilde{u}$ is uniformly distributed on $[-1,1]$ we get that
\begin{equation}\label{EabsOFuniform}
    \E\left[\left|\frac{\tilde{v}}{2\theta}+\tilde u\right|\right]
    =
    \frac{1}{2}\left(\frac{1}{\theta\sqrt{2\pi}}e^{-2\theta^2}+ \erf(\theta \sqrt{2}) \cdot \left(1+\frac{1}{4\theta^2}\right) \right). 
\end{equation}

Plugging it into \eqref{partial-c_uniform} we have
\begin{equation*}
\partial_{\lam} C(\theta,0) = -\frac{1}{\theta} +   \erf(\theta \sqrt{2}) \cdot \left(\theta+\frac{1}{4\theta}\right)+ \frac{1}{\sqrt{2\pi}}e^{-2\theta^2}.
\end{equation*}
It is easy to check that $\partial_{\lam} C(\theta,0)$ is monotone increasing (on $\theta >0$) and for $\theta > 1$ it is strictly positive.  

From \eqref{partial-c2} with $x(\tilde v)=\frac{\tilde{v}}{2\theta}$ we have
\begin{equation}\label{partial_theta_0}
\partial_{\theta} C(\theta,0)
 =  -2\E\left[ \tilde{v}\cdot \textrm{sign}\left(\frac{\tilde{v}}{2\theta}+\tilde u\right)\right]-\gamma \E\left[\left|\frac{\tilde{v}}{2\theta}+\tilde u \right|\right]+ 2\theta. 
\end{equation}

Note that 
\begin{align*}
    \E\left[ \tilde{v}\cdot \textrm{sign}\left(\frac{\tilde{v}}{2\theta}+\tilde u\right)\right]=&
   \int_{-1}^{1}\int_{-\infty}^{\infty}v (\mathds{1}_{\{v>-2u\theta\}}-\mathds{1}_{\{v<-2u\theta\}})\frac{1}{2} \frac{1}{\sqrt{2\pi}}e^{-v^2/2} dv du
    \\&=
    \frac{1}{\sqrt{2\pi}}\int_{-1}^{1}\int_{-2u\theta}^{\infty}v e^{-v^2/2} dv du
    \\&=
    \frac{1}{\sqrt{2\pi}}\int_{-1}^{1} e^{-2 \theta^2 u^2} du
    \\&=
    \frac{1}{2\theta}  \erf(\theta \sqrt{2}).
\end{align*}
Using this and \eqref{EabsOFuniform} in \eqref{partial_theta_0} we get 
\begin{equation*}
\partial_{\theta} C(\theta,0) = -\frac{1}{\theta}  \erf(\theta \sqrt{2})- \frac{\gamma}{2} \left(  \erf(\theta \sqrt{2}) \cdot \left(1+\frac{1}{4\theta^2}\right)+ \frac{1}{\theta\sqrt{2\pi}}e^{-2\theta^2}\right)+ 2\theta .
 \end{equation*}
Define 
 \begin{equation*}
 H(\theta)= -\frac{1}{\theta}  \erf(\theta \sqrt{2})- \frac{\gamma}{2} \left(  \erf(\theta \sqrt{2}) \cdot \left(1+\frac{1}{4\theta^2}\right)+ \frac{1}{\theta\sqrt{2\pi}}e^{-2\theta^2}\right)+ 2\theta .
 \end{equation*}
Note that $H$ is continuous and monotone increasing on $\theta>0$. Moreover $\lim_{\theta \rr 0} H(\theta) = -\infty $ and $\lim_{\theta \rr \infty}  H(\theta)= \infty$. It follows that $H$ has a unique zero $\theta^* = \theta^*(\gamma)$, which is clearly monotone increasing in $\gamma$ and $\lim_{\gamma \rr \infty } \theta^*(\gamma) = \infty$. 
We therefore showed that for any $\gamma >0$ we have $\theta^*>0$ such that \eqref{der-cond} holds.

Let $\eps>0$ be arbitrary small. Choose $\gamma$ large enough so that $\theta^*$ satisfies 
$$
P(|\tilde v|< \theta^* ) > 1-\eps. 
$$
Define $(x_0(v),\lam_0,\theta_0) = \left( \frac{v}{2\theta^*}, 0, \theta^*\right)$. From Proposition \ref{prop-unif} it follows that  $x_0(v) =  \frac{v}{2\theta^*}$  solves the insider optimisation problem in step 2 of Algorithm \ref{alg-kyle-bid-ask} if $|\tilde v| < \theta^* $. Moreover, since $\theta^*$ satisfies \eqref{der-cond} and $C(\lam,\theta)$ is convex, it follows that $(\lam_0,\theta_0)$ minimises $C(\lam,\theta)$, hence it is the output of step 3 in Algorithm \ref{alg-kyle-bid-ask}. We get that 
$$
(x_1(v),\lam_1,\theta_1) = \left( \frac{v}{2\theta^*}, 0, \theta^* \right), \quad \textrm{ with probability larger than } 1-\eps. 
$$
Repeating this argument we get \eqref{met-stable} for any $n\geq 1$. 
\end{proof}

\section{Formulas for equilibrium points} \label{sec-form}
In this section we derive simplified formulas for $\ell_{p,x^{*}}, \ p=1,2$, $\mu_{x^{*}} $ and $\kappa_{x^{*}}$ from Theorem \ref{thm-equil1}, for the case where $\tilde v-p_0$ is a standard Gaussian and $\tilde u$ is distributed uniformly on $[-1,1]$. We recall that in this case $x^*(v)$ is given by Proposition \ref{prop-unif}. Note that the expressions obtained for $\ell_{2,x^*}$ and $\mu_{x^*}$ are given closed form. The formulas for $\kappa_x^*$ and $\ell_{1,x^*}$ are given as an integral which could easily be evaluated by standard numerical schemes.  
We first introduce some notation. 
\paragraph{Notation.}
Recall that $\phi$ and $\Phi$ are the probability density function and cumulative distribution function of the standard Gaussian distribution, respectively. 

For any nonnegative $\lam$ and $\theta$ let 
$$
\beta(\lam,\theta)= \lam+\theta +\sqrt{\lam+\theta}. 
$$
For any integrable functions $f,g:\re \rr \re$ we define: 
\bn
F_1(\lam,\theta) &=&  \int_{0}^{\beta(\lam,\theta)}z^2\phi(z)dz= \frac{1}{2} \erf\left(\frac{\beta(\lam,\theta)}{2} \right)-\frac{\beta(\lam,\theta)}{\sqrt{2\pi}}e^{-\beta(\lam,\theta)^2/2}.  \\
F_2(\lam,\theta ; [f],[g]) &=&\int_{0}^{\infty} f(z) (1 \wedge g(z) +1)_{+} \mathds{1}_{\{z >\beta(\lam,\theta)\}} \phi(z)  dz, \\
\overline F_2(\lam,\theta ; [f],[g]) &=&\int_{0}^{\infty} f(z) (1 \wedge g(z) +1)_{+} \mathds{1}_{\{z \leq\beta(\lam,\theta)\}} \phi(z)  dz, \\
F_3(\lam, \theta; [f]) &=&\int_{0}^{\infty}(1-(1 \wedge f(z)))^2)\mathds{1}_{\{f(z)>-1\}} \mathds{1}_{\{z > \beta(\lam,\theta)\}} \phi(z)  dz \\
\overline F_3(\lam, \theta; [f]) &=&\int_{0}^{\infty}(1-(1 \wedge f(z)))^2)\mathds{1}_{\{f(z)>-1\}} \mathds{1}_{\{z \leq \beta(\lam,\theta)\}} \phi(z)  dz. 
\en

We start with the expression for $\ell_{2,x^{*}}$ 
\begin{lemma} \label{lem-l2} 
Under the assumptions of Proposition \ref{prop-unif} we have 
\bn
\ell_{2,x^{*}} &=& \frac{1}{2\lam^2}\left(\frac{1}{2}-F_1(\lam,\theta)-2\theta \frac{1}{\sqrt{2\pi}}e^{-\beta(\lam,\theta^2)^2/2} +\theta^2 (1-\Phi(\beta(\lam,\theta))  \right)\\
&&+ \frac{1}{2(\lam+\theta)^2} F_1(\lam,\theta). 
\en
\end{lemma} 

\begin{proof} 
From the independence of $\tilde v$ and $\tilde u$ we have 
\bn
\ell_{2,x^{*}} &=& \E\big[(x^*(\tilde v) +\tilde u)^2\big] \\ 
&=& \E[x^*(\tilde v)^2] +\sigma_{\tilde u}^2. 
\en
Recall that by Proposition \ref{prop-unif}, $x^*(v)$ is symmetric around $p_0$. Using the explicit formula for $x^*$ and since $\tilde v- p_0$ is a standard Gaussian we have 
\bn
 \E[x^*(\tilde v)^2]
&=&2\frac{1}{4(\lam+\theta)^2} \E\big[(\tilde v- p_0)^2\mathds{1}_{\{0\leq \tilde v-p_0 \leq \lam+\theta +\sqrt{\lam+\theta}\}} \big]\\
&&+ 2\frac{1}{4\lam^2} \E\big[(\tilde v-p_0-\theta)^2 \mathds{1}_{\{\tilde v-p_0 > \lam+\theta +\sqrt{\lam+\theta}\}}\big] \\
&=& \frac{1}{2\lam^2}\left(\frac{1}{2}-F_1(\lam,\theta)-2\theta \frac{1}{\sqrt{2\pi}}e^{-\beta(\lam,\theta)^2/2} +\theta^2 (1-\Phi(\beta(\lam,\theta)) \right)\\
&&+ \frac{1}{2(\lam+\theta)^2} F_1(\lam,\theta). 
 \en
 \end{proof} 

 Next, we derive an expression for $\mu_{x^{*}} $. 
 \begin{lemma} \label{lem-mu} 
Under the assumptions of Proposition \ref{prop-unif} we have 
$$
\mu_{x^{*}} =\frac{1}{(\lam+\theta)} F_1(\lam,\theta) + \frac{1}{\lam}\left(\frac{1}{2}-F_1(\lam,\theta) -\theta \frac{1}{\sqrt{2\pi}}e^{-\beta(\lam,\theta)^2/2} \right ).
$$ 
\end{lemma} 
 \begin{proof} 
 From Proposition \ref{prop-unif}, the symmetry of $x^*(v)$ around $p_0$ and since $\tilde v- p_0$ is a standard Gaussian we have 
 \bn
 \mu_{x^{*}} &=& \E\big[x^*(v)(\tilde v-p_0)\big] \\
 &=&2\frac{1}{2(\lam+\theta)} \E\big[(\tilde v- p_0)^2\mathds{1}_{\{0\leq \tilde v-p_0 \leq \lam+\theta +\sqrt{\lam+\theta}\}} \big]\\
&&+ 2\frac{1}{2\lam} \E\big[(\tilde v-p_0-\theta)(\tilde v- p_0) \mathds{1}_{\{\tilde v-p_0 > \lam+\theta +\sqrt{\lam+\theta}\}}\big] \\
&=& \frac{1}{(\lam+\theta)} F_1(\lam,\theta) + \frac{1}{\lam}\left(\frac{1}{2}-F_1(\lam,\theta) -\theta \frac{1}{\sqrt{2\pi}}e^{-\beta(\lam,\theta)^2/2} \right ).  \en
\end{proof} 
Next we compute $\kappa_{x^{*}}$. 
 \begin{lemma} \label{lem-kappa} 
Under the assumptions of Proposition \ref{prop-unif} we have 
$$
\kappa_{x^{*}} =2\overline F_{2}\Big(\lam,\theta;[z], \Big[\frac{z}{2(\lam+\theta)}\Big]\Big ) +  \left(F_{2}\Big(\lam,\theta;[z], \Big[\frac{z-\theta}{2\lam}\Big]\Big ) + F_{2}\Big(\lam,\theta;[z]\Big[\frac{z+\theta}{2\lam}\Big]\Big) \right) .
$$ 
\end{lemma} 
\begin{proof} 
Note that
\bn
\kappa_{x^{*}} &=& \E\big[\sign(x^*(\tilde v) +\tilde u)(\tilde v -p_0) \big] \\ 
&=&\E\big[(\tilde v -p_0) \mathds{1}_{\{x^*(\tilde v) +\tilde u \geq 0\}}\big] - \E\big[ (\tilde v -p_0) \mathds{1}_{\{x^*(\tilde v) +\tilde u <0\}}\big] \\ 
&=&:I_1-I_2. 
\en
We denote $(x)_{+} = \max\{0,x\}$. Using Proposition \ref{prop-unif} and the symmetry of $x^*(v)$ around $p_0$ we have  
\bn
I_1 &=& 2\E\big[ (\tilde v -p_0) \mathds{1}_{\{x^*(\tilde v) +\tilde u \geq 0\}}\mathds{1}_{\{0\leq \tilde v-p_0 \leq \lam+\theta +\sqrt{\lam+\theta}\}} \big] \\
&&+2\E\big[ (\tilde v -p_0) \mathds{1}_{\{x^*(\tilde v) +\tilde u \geq 0\}}\mathds{1}_{\{\tilde v-p_0 >  \lam+\theta +\sqrt{\lam+\theta}\}} \big] \\
&=& \int_{0}^{\infty} \int_{-1}^{1}z \mathds{1}_{\{u \geq -\frac{z}{2(\lam+\theta)}\}}\mathds{1}_{\{z \leq \lam+\theta +\sqrt{\lam+\theta}\}} \phi(z) du dz \\ 
&&+ \int_{0}^{\infty}\int_{-1}^{1} z \mathds{1}_{\{u \geq -\frac{z-\theta}{2\lam}\}}\mathds{1}_{\{z > \lam+\theta +\sqrt{\lam+\theta}\}} \phi(z) du dz \\
&=&\int_{0}^{\infty} z\Big(1-(-1) \vee \Big(-\frac{z}{2(\lam+\theta)}\Big)\Big)_{+} \mathds{1}_{\{z \leq \lam+\theta +\sqrt{\lam+\theta}\}} \phi(z)  dz \\ 
&&+\int_{0}^{\infty} z \Big(1-(-1)\vee \Big(- \frac{z-\theta}{2\lam}\Big)\Big)_{+} \mathds{1}_{\{z > \lam+\theta +\sqrt{\lam+\theta}\}} \phi(z)  dz \\
&=&\int_{0}^{\infty} z\Big(1+ 1 \wedge \Big(\frac{z}{2(\lam+\theta)}\Big) \Big)_{+} \mathds{1}_{\{z \leq \lam+\theta +\sqrt{\lam+\theta}\}} \phi(z)  dz \\ 
&&+\int_{0}^{\infty} z \Big(1+1\wedge \Big( \frac{z-\theta}{2\lam}\Big)\Big)_{+} \mathds{1}_{\{z > \lam+\theta +\sqrt{\lam+\theta}\}} \phi(z)  dz,
\en
where we have used the identity $(-x)\vee (-y) = -(x\wedge y)$ in the last iquality.  

On the other hand, 
\bn
I_2 &=& 2\E\big[(\tilde v -p_0) \mathds{1}_{\{x^*(\tilde v) +\tilde u < 0\}}\mathds{1}_{\{0\leq \tilde v-p_0 \leq \lam+\theta +\sqrt{\lam+\theta}\}} \big] \\
&&+2\E\big[(\tilde v -p_0) \mathds{1}_{\{x^*(\tilde v) +\tilde u < 0\}}\mathds{1}_{\{\tilde v-p_0 >  \lam+\theta +\sqrt{\lam+\theta}\}} \big] \\
&=& \int_{0}^{\infty} \int_{-1}^{1}z \mathds{1}_{\{u \leq -\frac{z}{2(\lam+\theta)}\}}\mathds{1}_{\{z \leq \lam+\theta +\sqrt{\lam+\theta}\}} \phi(z) du dz \\ 
&&+ \int_{0}^{\infty}\int_{-1}^{1} z \mathds{1}_{\{u \leq -\frac{z-\theta}{2\lam}\}}\mathds{1}_{\{z > \lam+\theta +\sqrt{\lam+\theta}\}} \phi(z) du dz \\
&=& \int_{0}^{\infty} z\Big(1 \wedge \Big(-\frac{z}{2(\lam+\theta)}\Big) +1\Big)_{+}  \mathds{1}_{\{z \leq \lam+\theta +\sqrt{\lam+\theta}\}} \phi(z) dz \\ 
&&+ \int_{0}^{\infty} z \Big(1 \wedge \Big(-\frac{z-\theta}{2\lam}\Big) +1\Big)_{+} \mathds{1}_{\{z > \lam+\theta +\sqrt{\lam+\theta}\}} \phi(z)  dz. 
\en
By a change of variable it follows that 
\bn
&&\kappa_{x^{*}} \\
&& = 2\int_{0}^{\infty} z\Big(1+ 1 \wedge \Big(\frac{z}{2(\lam+\theta)}\Big) \Big)_{+} \mathds{1}_{\{z \leq \lam+\theta +\sqrt{\lam+\theta}\}} \phi(z)  dz\\
&&+\int_{0}^{\infty} z \Big[ \Big(1 \wedge \Big(\frac{z-\theta}{2\lam}\Big) +1\Big)_{+}+\Big(1 \wedge \Big(\frac{z+\theta}{2\lam}\Big) +1\Big)_{+}\Big] \mathds{1}_{\{z > \lam+\theta +\sqrt{\lam+\theta}\}} \phi(z)  dz. 
\en
\end{proof} 

\begin{lemma} 
Under the assumptions of Proposition \ref{prop-unif} we have 
\bn
\ell_{1,x^*} &=&   2\overline F_{2}\Big(\lam,\theta;\Big[\frac{z }{2\lambda+\theta}\Big],\Big[\frac{z }{2\lambda+\theta}\Big] \Big) + \overline F_{3}\Big(\lam,\theta; \Big[\frac{z}{2(\lam+\theta)}\Big]\Big)\\
 &&+ F_{2}\Big(\lam,\theta;\Big[\frac{z-\theta }{2\lambda}\Big],\Big[\frac{z-\theta }{2\lambda}\Big] \Big ) + F_{2}\Big(\lam,\theta;\Big[\frac{z+\theta }{2\lambda}\Big],\Big[\frac{z+\theta }{2\lambda}\Big] \Big)   \\
 &&+\frac{1}{2} \left(F_{3}\Big(\lam,\theta; \Big[\frac{z-\theta}{2\lam}\Big]\Big)+F_{3}\Big(\lam,\theta;\Big[- \frac{z-\theta}{2\lam}\Big]\Big)  \right). 
\en
\end{lemma} 
\begin{proof} 
\bn
\ell_1(x^{*}) &=& \E\big[|x^*(\tilde v) +\tilde u|\big] \\ 
&=&\E\big[(x^*(\tilde v) +\tilde u) \mathds{1}_{\{x^*(\tilde v) +\tilde u \geq 0\}}\big] - \E\big[ (x^*(\tilde v) +\tilde u)\mathds{1}_{\{x^*(\tilde v) +\tilde u <0\}}\big] \\ 
&=&:I_1-I_2. 
\en
Using Proposition \ref{prop-unif} and the symmetry of $x^*(v)$ around $p_0$ we have  
\bn
I_1 &=& 2\E\big[ (x^*(\tilde v) +\tilde u)  \mathds{1}_{\{x^*(\tilde v) +\tilde u \geq 0\}}\mathds{1}_{\{0\leq \tilde v-p_0 \leq \lam+\theta +\sqrt{\lam+\theta}\}} \big] \\
&&+2\E\big[ (x^*(\tilde v) +\tilde u) \mathds{1}_{\{x^*(\tilde v) +\tilde u \geq 0\}}\mathds{1}_{\{\tilde v-p_0 >  \lam+\theta +\sqrt{\lam+\theta}\}} \big] \\
&=&\int_{0}^{\infty} \int_{-1}^{1}\Big( \frac{z}{2(\lam+\theta)} +u\Big)\mathds{1}_{\{u \geq -\frac{z}{2(\lam+\theta)}\}}\mathds{1}_{\{z \leq \lam+\theta +\sqrt{\lam+\theta}\}} \phi(z) du dz \\ 
&&+\int_{0}^{\infty}\int_{-1}^{1} \Big( \frac{z-\theta}{2\lam}+ u\Big)\mathds{1}_{\{u \geq -\frac{z-\theta}{2\lam}\}}\mathds{1}_{\{z > \lam+\theta +\sqrt{\lam+\theta}\}} \phi(z) du dz \\
&=& \int_{0}^{\infty} \frac{z}{2(\lambda+\theta)}\Big(1-(-1) \vee \Big(-\frac{z}{2(\lam+\theta)}\Big)\Big)_+\mathds{1}_{\{z \leq \lam+\theta +\sqrt{\lam+\theta}\}} \phi(z)  dz \\
&&+\frac{1}{2} \int_{0}^{\infty}\Big(1-\Big((-1) \vee \Big(-\frac{z}{2(\lam+\theta)}\Big)\Big)^2\Big)\mathds{1}_{\{-\frac{z}{2(\lam+\theta)}<1\}}
 \mathds{1}_{\{z \leq \lam+\theta +\sqrt{\lam+\theta}\}} \phi(z) du dz \\ 
 &&+ \int_{0}^{\infty} \frac{z-\theta }{2\lambda}\Big(1-(-1) \vee \Big(-\frac{z-\theta}{2\lam}\Big)\Big)_+\mathds{1}_{\{z > \lam+\theta +\sqrt{\lam+\theta}\}} \phi(z)  dz \\ 
&&+ \frac{1}{2}\int_{0}^{\infty}\Big(1-\Big((-1) \vee \Big(-\frac{z-\theta}{2\lam}\Big)\Big)^2\Big)\mathds{1}_{\{-\frac{z-\theta}{2\lam}<1\}}
 \mathds{1}_{\{z > \lam+\theta +\sqrt{\lam+\theta}\}} \phi(z) du dz.
\en
Since $(-x)\vee (-y) = -(x\wedge y)$ it follows that 
\bn
I_1 &=& \int_{0}^{\infty} \frac{z}{2(\lambda+\theta)}\Big(1 \wedge \Big(\frac{z}{2(\lam+\theta)}\Big)+1\Big)_+\mathds{1}_{\{z \leq \lam+\theta +\sqrt{\lam+\theta}\}} \phi(z)  dz \\ 
&&+ \frac{1}{2}\int_{-\infty}^{\infty}\Big(1-\Big(1 \wedge \Big(\frac{z}{2(\lam+\theta)}\Big)\Big)^2\Big)\mathds{1}_{\{-\frac{z}{2(\lam+\theta)}<1\}}
 \mathds{1}_{\{z \leq \lam+\theta +\sqrt{\lam+\theta}\}} \phi(z)  dz \\ 
 &&+\int_{0}^{\infty} \frac{z-\theta }{2\lambda}\Big(1+1 \wedge \Big(\frac{z-\theta}{2\lam}\Big)\Big)_+\mathds{1}_{\{z > \lam+\theta +\sqrt{\lam+\theta}\}} \phi(z)  dz \\ 
&&+ \frac{1}{2}\int_{0}^{\infty}\Big(1-\Big(1 \wedge \Big(\frac{z-\theta}{2\lam}\Big)\Big)^2\Big)\mathds{1}_{\{-\frac{z-\theta}{2\lam}<1\}}
 \mathds{1}_{\{z > \lam+\theta +\sqrt{\lam+\theta}\}} \phi(z)  dz. 
\en
On the other hand, 
\bn
I_2 &=& 2E\big( (x^*(\tilde v) +\tilde u)  \mathds{1}_{\{x^*(\tilde v) +\tilde u < 0\}}\mathds{1}_{\{0\leq \tilde v-p_0 \leq \lam+\theta +\sqrt{\lam+\theta}\}} \big) \\
&&+2E\big( (x^*(\tilde v) +\tilde u) \mathds{1}_{\{x^*(\tilde v) +\tilde u < 0\}}\mathds{1}_{\{\tilde v-p_0 >  \lam+\theta +\sqrt{\lam+\theta}\}} \big) \\
&=& \int_{0}^{\infty} \int_{-1}^{1}\Big( \frac{z}{2(\lam+\theta)} +u\Big)\mathds{1}_{\{u < -\frac{z}{2(\lam+\theta)}\}}\mathds{1}_{\{z \leq \lam+\theta +\sqrt{\lam+\theta}\}} \phi(z) du dz \\ 
&&+ \int_{0}^{\infty}\int_{-1}^{1} \Big( \frac{z-\theta}{2\lam}+ u\Big)\mathds{1}_{\{u < -\frac{z-\theta}{2\lam}\}}\mathds{1}_{\{z > \lam+\theta +\sqrt{\lam+\theta}\}} \phi(z) du dz \\
&=& \int_{0}^{\infty} \frac{z}{2(\lambda+\theta)}\Big(1 \wedge \Big(-\frac{z}{2(\lam+\theta)}\Big)+1\Big)_+\mathds{1}_{\{z \leq \lam+\theta +\sqrt{\lam+\theta}\}} \phi(z)  dz \\ 
&&+ \frac{1}{2}\int_{0}^{\infty}\Big(1 \wedge \Big(-\frac{z}{2(\lam+\theta)}\Big)\Big)^2-1\Big)\mathds{1}_{\{-\frac{z}{2(\lam+\theta)}>-1\}}
 \mathds{1}_{\{z \leq \lam+\theta +\sqrt{\lam+\theta}\}} \phi(z) dz\\ 
 &&+\int_{0}^{\infty} \frac{z-\theta }{2\lambda}\Big(1 \wedge \Big(-\frac{z-\theta}{2\lam}\Big)+1\Big)_+\mathds{1}_{\{z > \lam+\theta +\sqrt{\lam+\theta}\}} \phi(z)  dz \\ 
&&+ \frac{1}{2}\int_{0}^{\infty}\Big(\Big(1 \wedge \Big(-\frac{z-\theta}{2\lam}\Big)\Big)^2-1\Big)\mathds{1}_{\{-\frac{z-\theta}{2\lam}>-1\}}
 \mathds{1}_{\{z > \lam+\theta +\sqrt{\lam+\theta}\}} \phi(z)  dz. 
\en
By a change of variable we get, 
\bn
I_2 &=&- \int_{0}^{\infty} \frac{z}{2(\lambda+\theta)}\Big(1 \wedge \Big(\frac{z}{2(\lam+\theta)}\Big)+1\Big)_+\mathds{1}_{\{z \leq \lam+\theta +\sqrt{\lam+\theta}\}} \phi(z)  dz \\ 
&&+\frac{1}{2}\int_{0}^{\infty}\Big(\Big(1 \wedge \Big(\frac{z}{2(\lam+\theta)}\Big)\Big)^2-1\Big)\mathds{1}_{\{\frac{z}{2(\lam+\theta)}>-1\}}
 \mathds{1}_{\{z \leq \lam+\theta +\sqrt{\lam+\theta}\}} \phi(z)  dz \\ 
 &&+ \int_{0}^{\infty} \frac{z-\theta }{2\lambda}\Big(1 \wedge \Big(-\frac{z-\theta}{2\lam}\Big)+1\Big)_+\mathds{1}_{\{z > \lam+\theta +\sqrt{\lam+\theta}\}} \phi(z)  dz \\ 
&&+ \frac{1}{2}\int_{0}^{\infty}\Big(\Big(1 \wedge \Big(-\frac{z-\theta}{2\lam}\Big)\Big)^2-1\Big)\mathds{1}_{\{-\frac{z-\theta}{2\lam}>-1\}}
 \mathds{1}_{\{z > \lam+\theta +\sqrt{\lam+\theta}\}} \phi(z)  dz. 
\en
It follows that 
\bn
&&\ell_{1,x^{*}} \\  
&&= I_1 -I_2  \\
&&=2\int_{0}^{\infty} \frac{z}{2(\lambda+\theta)}\Big(1 \wedge \Big(\frac{z}{2(\lam+\theta)}\Big)+1\Big)_+\mathds{1}_{\{z \leq \lam+\theta +\sqrt{\lam+\theta}\}} \phi(z)  dz \\
&&+ \int_{0}^{\infty}\Big(1-\Big(1 \wedge \Big(\frac{z}{2(\lam+\theta)}\Big)\Big)^2\Big)\mathds{1}_{\{\frac{z}{2(\lam+\theta)}>-1\}}
 \mathds{1}_{\{z \leq \lam+\theta +\sqrt{\lam+\theta}\}} \phi(z)  dz\\
&&+ \int_{0}^{\infty} \Big[\frac{z-\theta }{2\lambda}\Big(1+1 \wedge \Big(\frac{z-\theta}{2\lam}\Big)\Big)_+ + \frac{z+\theta }{2\lambda}\Big(1 \wedge \Big(\frac{z+\theta}{2\lam}\Big)+1\Big)_+\Big]\\
 && \qquad  \times \mathds{1}_{\{z > \lam+\theta +\sqrt{\lam+\theta}\}} \phi(z)  dz \\ 
&&+ \frac{1}{2}\int_{0}^{\infty}\Big(1-\Big(1 \wedge \Big(\frac{z-\theta}{2\lam}\Big)\Big)^2\Big)\mathds{1}_{\{\frac{z-\theta}{2\lam}>-1\}}
 \mathds{1}_{\{z > \lam+\theta +\sqrt{\lam+\theta}\}} \phi(z)  dz \\
 &&+ \frac{1}{2}\int_{0}^{\infty}\Big(1-\Big(1 \wedge \Big(-\frac{z-\theta}{2\lam}\Big)\Big)^2\Big)\mathds{1}_{\{-\frac{z-\theta}{2\lam}>-1\}}
 \mathds{1}_{\{z > \lam+\theta +\sqrt{\lam+\theta}\}} \phi(z)  dz \\
 &=& 2 \overline F_{2}\Big(\lam,\theta;\Big[\frac{z }{2\lambda+\theta}\Big],\Big[\frac{z }{2\lambda+\theta}\Big] \Big) +\frac{1}{2}\overline F_{3}\Big(\lam,\theta; \Big[\frac{z}{2(\lam+\theta)}\Big]\Big)\\
 &&+ \Big[F_{2}\Big(\lam,\theta;\Big[\frac{z-\theta }{2\lambda}\Big],\Big[\frac{z-\theta }{2\lambda}\Big] \Big ) + F_{2}\Big(\lam,\theta;\Big[\frac{z+\theta }{2\lambda}\Big],\Big[\frac{z+\theta }{2\lambda}\Big] \Big) \Big] \\
 &&+\frac{1}{2} \Big[F_{3}\Big(\lam,\theta; \Big[\frac{z-\theta}{2\lam}\Big]\Big)+F_{3}\Big(\lam,\theta;\Big[- \frac{z-\theta}{2\lam}\Big]\Big)  \Big]. 
\en
\end{proof}

\bigskip
\bibliographystyle{plain}

\end{document}